\documentclass[11pt]{article}
\pdfoutput=1 
\usepackage[sc]{mathpazo}

\usepackage[margin=1in]{geometry}

\usepackage{times}
\usepackage{soul}
\usepackage{url}
\usepackage[hidelinks]{hyperref}
\usepackage[utf8]{inputenc}
\usepackage[small]{caption}
\usepackage{graphicx}
\usepackage{amsmath}
\usepackage{amsthm}
\usepackage{booktabs}
\usepackage{algorithm}
\usepackage[noend]{algpseudocode}
\usepackage[switch]{lineno}
\usepackage{comment}

\usepackage[table]{xcolor}
\usepackage{tikz, pgfplots}
\usetikzlibrary{positioning}
\usepackage{tikz-network}
\usepackage{amssymb,amsthm,amsmath}
\usepackage{mathtools} 
\usepackage{array} 
\usepackage{subcaption}
\usepackage{float}

\newcommand{\hide}[1]{}
\usepackage{xspace}
\newcommand{\pallavi}[1]{{\color{blue}{Pallavi says: #1}}}

\newcommand{\shubham}[1]{{\color{red}{Shubham says: #1}}}
\newcommand{\mefem}{\textsc{MEFE}\xspace}

\newcommand{\dmatching}{\textsc{$3$-dimensional Perfect Matching}\xspace}

 \newcommand{\dmatchingshort}{\textsc{$3$D-PM}\xspace}

  \newcommand{\tcm}{\textsc{(3-3)-com smti}\xspace}

\newcommand{\ssmtiffull}{\textsc{Weighted Strongly Stable Matching with Ties and Incomplete Lists}\xspace}

\newcommand{\ssmtif}{\textsc{WSSMTI}\xspace}

\newcommand{\ssmtifull}{\textsc{Strongly Stable Matching with Ties and Incomplete Lists}\xspace}

\newcommand{\ssmti}{\textsc{SSMTI}\xspace}

\newcommand{\mefe}{\textsc{MEFE-Matching}\xspace}

\newcommand{\hrpfull}{\textsc{Hospitals Residents problem}\xspace}

\newcommand{\hrp}{\textsc{HR}\xspace}

\usepackage{todonotes}
\usepackage{amsfonts}

\newtheorem{corollary}{Corollary}
\newtheorem{lemma}{Lemma}

 \newtheorem{reduction rule}{Reduction Rule}

\newcommand{\nph}{\textsf{NP-hard}\xspace}
\newcommand{\npc}{\textsf{NP-complete}\xspace}
\newcommand{\fpt}{\textsf{FPT}\xspace}
\newcommand{\OO}{\mathcal{O}}

\usepackage{thmtools,thm-restate}
\usepackage{nicefrac}

  \usepackage{cleveref}

\usepackage{natbib}
\newtheorem{theorem}{Theorem}

\title{Fairness and Efficiency in Two-Sided Matching Markets}
\author{
Pallavi Jain, Palash Jha, and Shubham Solanki\\
Indian Institute of Technology Jodhpur\\
\{pallavi,jha.16,solanki.4\}@iitj.ac.in
}

\date{\vspace{-2em}}

\begin{document}

\maketitle

\begin{abstract}
    We propose a new fairness notion, motivated by the practical challenge of allocating teaching assistants (TAs) to courses in a department. Each course requires a certain number of TAs and each TA has preferences over the courses they want to assist. Similarly, each course instructor has preferences over the TAs who applied for their course. We demand fairness and efficiency for both sides separately, giving rise to the following criteria: (i) every course gets the required number of TAs and the average utility of the assigned TAs meets a threshold; (ii) the allocation of courses to TAs is envy-free, where a TA envies another TA if the former prefers the latter's course and has a higher or equal grade in that course. Note that the definition of envy-freeness here differs from the one in the literature, and we call it merit-based envy-freeness.
    We show that the problem of finding a merit-based envy-free and efficient matching is NP-hard even for very restricted settings, such as two courses and uniform valuations; constant degree, constant capacity of TAs for every course, valuations in the range $\{0,1,2,3\}$, identical valuations from TAs, and even more. To find tractable results, we consider some restricted instances, such as, strict valuation of TAs for courses, the difference between the number of positively valued TAs for a course and the capacity, the number of positively valued TAs/courses, types of valuation functions, and obtained some polynomial-time solvable cases, showing the contrast with intractable results. We further studied the problem in the paradigm of parameterized algorithms and designed some exact and approximation algorithms. 
\end{abstract}

\section{Introduction}

Assigning teaching assistants (TAs) to courses is a practical challenge faced by many academic departments. The allocation should be fair for both instructors and TAs,  taking into account their preferences and qualifications.  
For example, if a TA is willing to assist in a systems course, it is not desirable to assign them a theory course. On the contrary, it may also happen that a TA might like to assist in a course, however he does not have enough knowledge in that course. Considering some of these critical requirements of instructors and TAs, our department\footnote{we skip the name in this version in the interest of anonymity.} follows the following protocol: we ask TAs to submit their preferences over the courses they would like to assist with and their grades in those courses, and then instructors are asked to submit their preferences over the TAs who are interested in their courses. We try to ensure that every course gets the required number of TAs (at least some of the choice of instructor) and prioritise the choice of the TAs with higher grades. However, finding a satisfactory allocation that meets the requirements of both sides is not trivial. In this paper, we address this problem from a theoretical perspective, using tools from algorithms and computational social choice theory.

\paragraph{\bf Mathematical Formulation.}
We have a set of courses $X=\{x_1,\ldots,x_n\}$ and a set of TAs $T=\{t_1,\ldots,t_m\}$. Every TA has a utility for every course, i.e., for every $i\in [m]$, we have a utility function $u_i \colon X \rightarrow \mathbb{Z}_{\geq 0}$. Furthermore, every TA $t_i$ has a grade function $g_i\colon X \rightarrow \mathbb{Q}_{\geq 0}$. For simplicity, we assumed that every TA has done all the courses (grade can be considered $0$ if a TA has not credited a course). 
Every course $x_i\in X$ has a capacity $c_i$ (that captures the required number of TAs), and a utility function $v_i\colon T\rightarrow \mathbb{Z}_{\geq 0}$. If a TA $t$ has utility $0$ for the course $x_i$, then  $v_i(t)=0$, otherwise $v_i(t)>0$. This is a valid assumption, as we would like to respect the choices of a TA. A matching $\mu \colon T \rightarrow X \cup \{\emptyset\}$ is \emph{feasible} if for every course $x_i\in X$, $|\mu^{-1}(x_i)|=c_i$, i.e., $c_i$ TAs are assigned to the course $x_i$, and no TA is assigned to a zero-valued course. We call that a matching $\mu \colon T \rightarrow X \cup \{\emptyset\}$  is \emph{fair} if it is feasible and for a given $k\in \mathbb{Q}_{\geq 0}$, it satisfies the following two conditions.
\begin{itemize}
    \item {\bf Satisfaction of Courses.}\footnote{This is efficiency criteria. But, we include it in the definition of fair matching for brevity.} For every course $x_i\in X$, ${\sf AvgUtil}(x_i)=\frac{\sum_{t\in \mu^{-1}(x_i)}v_i(t)}{c_i}$ is at least $k$.
    \item {\bf Merit-based envy-freeness between TAs.} 
    A TA $t_i$ envies another TA $t_j$ on merit basis if $g_{i}(\mu(t_j))\geq g_{j}(\mu(t_j))$ and $u_i(\mu(t_j))>u_i(\mu(t_i))$, i.e., grade of $t_i$ in the course allocated to $t_j$ is at least the grade of $t_j$ in that course and $t_i$ values the course allocated to $t_j$ more than their own allocated course. In the matching $\mu$, there should not be any pair of merit-based envious TAs. 
\end{itemize}

We call this fair matching as \emph{merit-based envy-free egalitarian} (\mefem) matching, and the problem of finding such a matching as \mefe.  

\emph{Highest grade vs highest utility in a course.} It seems desirable to assume that the highest grader in a course will have the highest utility. However, this need not be the case. %
 
There might be
 other responsible factors for course utilities. For example, if a TA applicant has previously done a relevant course or project under the course instructor who trusts the TA’s expertise in the area and soft skills, the instructor would probably assign a higher utility for that TA. A TA applicant with prior experience at TAship or a recommendation from other instructors might also be assigned higher utility due to their
 experience. While TAs might have done the course previously, they may have done it in different universities. In such a case, it is hard to gauge a TA accurately just based on their grades. 
Thus, we demand utility functions from the courses separately.

\paragraph{\bf Significance of our Fairness Notions.} The deferred-acceptance algorithm (commonly-used algorithm in two-sided markets) is biased towards one side of the agents\sloppy~\cite{GS62college}, which is not desirable for our application. Thus, we move our attention to some of the well-studied fairness criterion used in the field of \emph{fair allocation}.

On the course side, we considered criteria similar to \emph{egalitarian welfare}, but we take the average for the following reason: suppose that there is a course $x$ with capacity $10$ and a course $y$ with capacity $1$. Now, if we require a threshold for satisfaction, say $10$, then the course $y$ gets a very good TA; however, course $x$ may get TAs each with utility $1$. To avoid such a predicament, we considered average. Since we require good TAs in all the courses, we think \emph{egalitarian type} criteria is a better choice over other welfare functions such as utilitarian or Nash. We believe that envy-freeness (the most popular fairness criteria in fair allocation) is not suitable here for the following reason: what if everybody gets a bad set of TAs? Nobody may envy anybody, but this is a disastrous situation for everyone. We also believe that instructors care more about their TAs instead of comparing them with the TAs of other courses.

On the TAs side, we considered \emph{envy-freeness type} fairness. As discussed earlier, it is important to consider grades as well. Our example in Figure~\ref{fig:example2} demonstrates that utility from the course side alone is not enough to prioritise the preferences of meritorious TAs. Indeed if a matching is envy-free, then it is also merit-based envy-free. However, it is possible that there exists a merit-based envy-free matching even if envy-free matching does not exist. Consider a classical example of non-existence of envy-free solution, when there is only one course $x$ with capacity $1$, and two TAs $t,t'$, both value $x$ positively. This does not have envy-free solution (irrespective of actual valuations). But, if $x$ values $t$ (say $\geq k$) more than $t'$ and the grade of $t$ is more than $t'$, then we match $t$ with $x$, and it is a merit-based envy-free solution.  As discussed earlier, conceptually also, it is crucial to consider grades.
\begin{figure}[]
    \centering

\begin{tikzpicture}

\node[draw] at (-1.5,4.7) {
($t_1,t_2,t_3$)};

\node[draw] at (3.4,4.7) {
($c_1,c_2,c_3$)};

\Vertex[x=0,y=4,size=0.2,label = ${(9,8,7)} \:  c_1$,fontscale =1.5,position = left]{l1};
\Vertex[x=0,y=3,size=0.2,label = ${(8,7,9)} \: c_2$,fontscale =1.5,position = left]{l2};
\Vertex[x=0,y=2,size=0.2,label = ${(7,7,7)} \: c_3$,fontscale =1.5,position = left]{l3};

\Vertex[x=2,y=4,size=0.2,label = $t_1 \: {(9,8,8)}$,fontscale =1.5,position = right]{r1};
\Vertex[x=2,y=3,size=0.2,label = $t_2 \: {(8,8,8)}$,fontscale =1.5,position = right]{r2};
\Vertex[x=2,y=2,size=0.2,label = $t_3 \: {(8,8,8)}$,fontscale =1.5,position = right]{r3};

\Edge(l3)(r3)
\Edge(l1)(r2)
\Edge(l2)(r1)

\end{tikzpicture}

    \caption{An instance to show unfairness on the TA's side if we consider grades as utilities from the course side and our satisfaction criteria of courses alone for a fair matching. Here the tuple $(p,q,r)$ for a course $c_i$ (or TA $t_i$), denotes the valuation for TAs $t_1,t_2,t_3$ (courses $c_1,c_2,c_3$), respectively. Edges denote a matching for $k=7$. TA $t_1$ has a higher grade in $c_1$ than TA $t_2$ and also values $c_1$ more than $c_2$. Thus, this matching is unfair for $t_1$.}
    \label{fig:example2}
\end{figure}
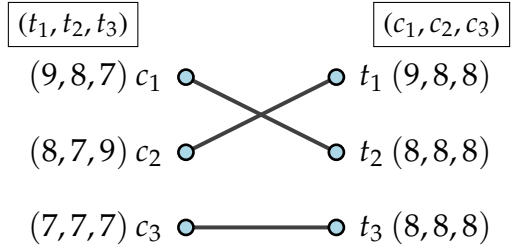

Since grades need to be considered, egalitarian criteria is not interesting for TAs.

Interestingly, \mefem matching is also a weakly stable matching when course satisfaction is not under consideration or utility of course $x$ for TA $t$ is same as the grade of $t$ in the course $x$, however the converse need not be true. 
For contradiction, we assume that the \mefe solution is not weakly stable. Let the blocking pair be $(x_i , t_j)$. Then $t_j$ prefers $x_i$ over $\mu(t_j)$ and has a higher grade in $x_i$ then $\mu^{-1}(x_i)$. Hence $t_j$ envies $\mu^{-1}(x_i)$, a contradiction. The converse may not be true. For example, consider an instance with 1 course and 2 TAs both having the same grade in the course, with the course having capacity 1. Any assignment will be weakly stable as the course is indifferent between the two TAs but the unassigned TA will always envy the assigned TA. Additionally, if the course values both TAs less than $k$, no assignment would lead to satisfaction of the course.

Indeed, there are several notions of fairness in the fair allocation literature, which will be interesting to consider, and some are already considered (see Related Work).  

\paragraph{\bf Our Contribution.} 

Our model considers a fairness criterion for TAs and an efficiency criterion for courses. This is motivated by the fact that the requirements of both sides in matching under two-sided preferences need not be the same. This is our first conceptual contribution. 

Next, we move to our technical contribution. Table~\ref{table:our results} summarises our complexity results. Some of the technical highlights of our work are discussed
below. Throughout the paper, by \emph{degree} of a course, we mean the number of TAs positively valued by the course. Similarly, we define the degree of a TA. 

\paragraph{\bf Existence of Solution.} In Section~\ref{sec:existence}, we identify two yes-instances and some no-instances of the problem.

\paragraph{\bf{Hardness Results}.} Unsurprisingly, similar to many problems in Fair Allocation, the problem is \npc for two courses ({\bf Theorem ~\ref{thm:nph-2courses}}), due to a reduction from the {\sc Equal-Cardinality Partition} problem. In this reduction, the degree of courses is “large”—that is, each course positively values many TAs—while each TA positively values only two courses (has degree 2). Furthermore both the courses have same valuation functions and all the TAs equally value all the courses, in fact, their grade is also same in all the courses. We next ask the question: \emph{does the large degree of courses leads to intractability?} Surprisingly, the problem is \npc even when the degree of courses is three ({\bf Theorem~\ref{thm:nph-cap1}}). In the same reduction, the degree of TAs is also constant, in particular, three. In fact, the capacity of all the courses is one and each course values all positively valued TAs equally (refer to the theorem statement for more restrictions). In light of Theorem~\ref{thm:nph-cap1}, we next ask: \emph{do the different valuations of TAs lead to intractability}? {\bf Theorem~\ref{thm:nph-cap2}} answers this question negatively. The next set of results, highlights more about the inherent nature of instance that leads to intractability. 

\paragraph{\bf{Polynomial Time Algorithms}} We observe that the large degree of courses is not responsible for the intractability; rather, the difference between degree and capacity is one of the responsible factors. In Theorem~\ref{thm:nph-cap1}, the difference between degree and capacity is two. In {\bf Theorem~\ref{degree-cap1}}, we show that the problem is polynomial-time solvable when this difference is at most one (the actual value of degree and capacity does not matter). This result also resolves the complexity when the degree of courses is two ({\bf Corollary~\ref{cor:deg2}}), contrasting Theorem~\ref{thm:nph-cap1}. We further observed that when the capacity of each course is one, then the same valuations of a TA for multiple courses lead to intractability, which is easy to avoid, but if all TAs value positively valued courses differently, then it leads to tractability ({\bf Theorem~\ref{poly-cap1}}). Due to Theorem~\ref{thm:nph-2courses}, we have hardness when the number of courses is two and the degree of TAs is two. When we relax one of these constraints, i.e., the degree of TAs is one or there is only one course, it leads to tractability ({\bf Theorem \ref{thm:poly-degree1} and Corollary~\ref{cor:deg1}}). In all our hardness results, either the capacity of courses is constant, or the number of courses is constant. When both are constant, then it leads to a tractability result ({\bf Theorem~\ref{poly-constant-course-cap}}). Contrasting Theorem~\ref{thm:nph-cap2}, {\bf Theorem~\ref{thm:2val}} shows that the problem is polynomial-time solvable when every course has at most two distinct positive valuations for all TAs, but, here we require that for a course no two TAs have the same grade (which is not a very strict restriction if we consider actual marks). 

\paragraph{{\bf Parameterized (Approximation) Algorithms.}} \mefe has some natural parameters to consider: the number of TAs ($m$), then number of courses ($n$), the maximum degree of courses/TAs (${\sf d_{course}, d_{TA}}$), maximum capacity of a course (${\sf cap}$), the number of distinct valuation/grade functions (${\sf type_{val}}, {\sf type_{grade}}$), maximum value that a valuation/grade functions can take (${\sf max_{val}}, {\sf max_{grade}}$). {\bf Theorem~\ref{thm:FPTm}} shows positive results with respect to $m$. Due to Theorem~\ref{thm:nph-cap1}, we have {\sf paraNP-hardness} with respect to ${\sf d_{course}}+d_{TA}+{\sf type_{val}}+{\sf type_{grade}}+{\sf max_{val}}+{\sf max_{grade}}+{\sf cap}$. Thus, we cannot hope for a \emph{fixed-parameter tractable}(\fpt) algorithm with respect to the combination of these parameters. We design an \fpt algorithm with respect to $n$ when ${\sf type_{val}}$ and ${\sf cap}$ are constant for courses, no TA values two courses equally and no course has same grade for two TAs ({\bf Theorem~\ref{thm:FPTn}}). When we relax the constraint of constant ${\sf type_{val}}$ (i.e., it is no longer constant) and ${\sf max_{val}}$ is a function of $n$, then we obtain a $(1-\epsilon)$-approximation algorithm that approximates satisfaction of every course and runs in $\fpt(n,\epsilon)$ ({\bf Theorem~\ref{thm:fpt-apxn}}). Note that due to Theorem~\ref{thm:nph-2courses},  we cannot hope for an \fpt algorithm with respect to $n+{\sf degree_{TA}+ type_{val}+type_{grade}}$

\begin{table*}
\begin{center}
\resizebox{\textwidth}{!}{%
\begin{tabular}{ |c|c|c|c|c|c| } 
 \hline
  $\#$ courses & Capacity & Degree & Types of valuations& Result & Theorem \\ 
\hline

\rowcolor{green!10} 2 & - & of TA = 2 & by TA = 1 & NP-hard&[\ref{thm:nph-2courses}]\\

 \hline
\rowcolor{green!10} - & 1 & of TA / course $\leq 3$ & by course = 1  & NP-hard& [\ref{thm:nph-cap1}]\\

 \hline
\rowcolor{red!10} - & 1 & - & -  & polynomial& [\ref{poly-cap1}]\\

\hline
\rowcolor{green!10} - & 2 & of TA $\leq 3$, course $\leq 6$ & by course = 3,by TA = 1  & NP-hard&[\ref{thm:nph-cap2}] \\ 
  
  \hline
\rowcolor{green!10} - & c & of course $\leq c+1$ & - & polynomial& [\ref{degree-cap1}]\\

\hline
\rowcolor{green!10} - & - & of TA = 1 & - & polynomial& [\ref{thm:poly-degree1}]\\ 

 \hline
\rowcolor{red!10} - & - & - &  by course = 2 & \textcolor{blue}{polynomial}& [\ref{thm:2val}]\\ 

  \hline
\rowcolor{green!10}  constant & constant & - & - &  polynomial&[\ref{poly-constant-course-cap}]\\ 
 
  \hline
\rowcolor{red!10} - & constant & - & by course = constant & \textcolor{blue}{$\fpt(n)$}&[\ref{thm:FPTn}]\\ 
  \hline
\rowcolor{green!10} - & - & - & - & $\fpt(m)$ &[\ref{thm:FPTm}]\\ 
  \hline
\rowcolor{red!10} -& constant & - & - & \textcolor{blue}{$(1-\epsilon)$- approximation in $\fpt(n,\epsilon)$}&[\ref{thm:fpt-apxn}]\\ 
  \hline

\end{tabular}%
}
\caption{Summary of results where color denotes the ties in TA's preference list {\color{green!70} green} means ties are allowed, \textcolor{red!70}{red} means ties are not allowed. Text written in \textcolor{blue}{blue} denotes a condition that in a course no two TAs should have the same grade. The number of courses and TAs are denoted by $n$ and $m$, respectively. ``-'' means the values can be arbitrary.   
{\bf In the approximation result (last row), ${\sf max_{val}}$ is a function of $n$}} 
\label{table:our results}
\end{center}
\end{table*}

\paragraph{\bf Related Work.}
The literature on matching under two-sided preferences is vast.
It originated from the seminal work of \cite{GS62college} on the stable matching problem. The unfairness of the deferred-acceptance algorithm, together with the immense practical applicability of stable matchings, has generated considerable interest in developing algorithms for finding \emph{fair} stable matchings. Several fairness concepts have been studied in conjunction with stability, including \emph{minimum regret}~\cite{K97stable}, \emph{egalitarian}~\cite{MW71stable,ILG87efficient},\footnote{In the stable matching literature, the term \emph{egalitarian matching} has been used for matchings that maximize the total satisfaction of the agents by minimizing the \emph{sum} of ranks of the matched partners~\cite{MW71stable,ILG87efficient}. This objective is different from \emph{egalitarian welfare}, which maximizes the utility of the least-happy agent, which we considered.} \emph{median}~\cite{TS98geometry,STQ06many}, \emph{sex-equal}~\cite{GI89stable}, \emph{balanced}~\cite{F95stable,GRS+21balanced},  \emph{leximin}~\cite{NBN22achieving}, and Nash~\cite{jain2024maximizing}. 

For the related work on matching on one-sided preferences (a field of Fair Allocation), we refer the reader to recent surveys~\cite{DBLP:journals/ai/AmanatidisABFLMVW23,DBLP:conf/aaai/LiuLSW24,DBLP:journals/corr/abs-2307-10985}. 

Freeman et al.~\cite{DBLP:conf/ijcai/FreemanM021} consider envy-freeness up to one good and maximin share guarantee under two sided-preferences. But, in their model, the fairness criteria are the same on both sides. However, this is the first work (to the best of our knowledge) that borrows fairness notions from the fair division literature to the two-sided preferences, conceptually closer to our work. Bu et al.~\cite{bu2023fair} extended this work and considered the fairness criterion envy-freenes up to $c$-good and proportionality up to $c$-good.  They also considered the same fairness criteria on both the sides. Patro et al.~\cite{patro2020fairrec} considered envy-freeness up to one individual on one side and maximin share on the other side. Gollapudi et al.~\cite{gollapudi2020almost} also considered envy-freeness up to one individual on both the sides along with the maximum weight of the matching in case of repeated matchings.    Igarashi et al.~\cite{DBLP:conf/ijcai/0001KSS23} considered stability condition along with envy-free up to one individual for one side of the agents. Recently, Sung-Ho Cho~\cite{DBLP:conf/atal/ChoKLLL0YY24} studied the problem from a mechanism design perspective and considered envy-freeness up to $k$-peers. 
There are several other papers that study fairness in two-sided preferences (not directly related to our work)~\cite{DBLP:conf/innovations/KimKRY20,DBLP:conf/innovations/DworkHPRZ12,DBLP:conf/innovations/KarniRY22,bandyapadhyay2023proportionally,gupta2023towards}. Indeed the list is not exhaustive.

Bredereck et al.~\cite{DBLP:conf/atal/Bredereck0N18} also studied local envy-freeness in house allocation problems where each agent receives only one item. Our model generalizes local envy-freeness as follows. We create a directed graph for each course, where TA $t$ has an arc to TA $t’$ if $t$ has a grade at least as high as $t'$ in that course. We seek an envy-free allocation from the TAs' perspective that meets this condition: if course $x$ is assigned to $t$, let $Z$ be the set of courses $t$ values more than $x$. We then check the local envy-free condition for $t$ in all graphs corresponding to the courses in $Z$.

\section{Preliminaries}

Throughout the paper, $(X,T,\{v_i\}_{i\in X}, \{u_i\}_{i\in T}, \{g_i\}_{i\in T},$ $\{c_i\}_{i\in X},k)$ denotes an instance of \mefe, where $X=\{x_1,\ldots,x_n\}$ and $T=\{t_1,\ldots,t_m\}$. When referring to an element $t \in T$, without a subscript, we may use grade function $g_t : X \rightarrow \mathbb{Q}_{\geq 0}$, and utility function $u_t : X \rightarrow \mathbb{Z}_{\geq 0}$, both with subscript $t$. Similarly, when referring to a course $x \in X$, without a subscript, we may use utility function $v_x : T \rightarrow \mathbb{Z}_{\geq 0}$. By \emph{degree} of a course, we mean the number of TAs positively valued by the course. Similarly, we define the degree of a TA. We denote the degree of course/TA $z$ by $d(z)$. For a course $x\in X$, $N(x)$ is the set of all TAs positively valued by $x$. Similarly, for a TA $t\in T$, $N(t)$ is the set of all courses positively valued by $t$. We define a similar notion for $Z\subseteq X$ (or $Z\subseteq T$) as follows: $N(Z)=\cup_{z\in Z}N(z)$. 
By \emph{types of valuations} of a course/TA $z$, we mean the number of distinct positive values assigned by $z$ to TAs/courses. The functions $f\colon P\rightarrow R$ and $g\colon P\rightarrow R$ are \emph{identical}, if for all $z\in P$, $f(z)=g(z)$. A function whose range is  $\{0,1\}$ is called \emph{binary}. For $z\in X\uplus T$, we  also use the notation $u_z$ for the utility or grade function of $z$. 

Let $\mu$ be a matching. If a course (TA) is not matched to a TA (course) in $T$ ($X$), then we call it \emph{unassigned} or \emph{unsaturated} with respect to $\mu$; otherwise it is called  \emph{assigned} or \emph{saturated}. 

For an instance ${\cal I}$ of \mefe, let $G_{\cal I}=(X,T)$ be a bipartite graph such that $E(G_{\cal I})=\{(x,t)\colon u_t(x)\neq 0\}$. Note that $G_{\cal I}$ can be a disconnected graph. We skip the subscript ${\cal I}$, if the instance is clear from the context.

In parameterized algorithms, given an instance $I$ of the problem $\Pi$, and an integer $k$ (called a parameter), the goal is to design an algorithm that runs in $f(k)\cdot |I|^{\OO(1)}$ time, where $f$ is an
arbitrary computable function depending on the parameter $k$. Such algorithms are known as \fpt algorithms. If there exists an \fpt algorithm with respect to the parameter $k$ for the problem $\Pi$, we say that ``$\Pi$ can be solved in $\fpt(k)$''.
For more details on the subject, 
we refer to the textbooks~\cite{ParamAlgorithms15b,fg,downey}.

\section{Hardness Results}
In this section, we prove the intractability of the problem even for very restricted cases. 
Our first intractability result is for two courses with identical valuations, where degree of TAs is two, contrasting Theorem~\ref{thm:poly-degree1} and Corollary~\ref{cor:deg1}.
The result is due to the polynomial-time reduction from the {\sc Equal-Cardinality Partition} problem, in which given a multiset ${\cal S}$ of positive integers, the goal is to partition ${\cal S}$ into two subsets of the same
size that have the same sum. This problem is known to be \npc~\cite{DBLP:books/fm/GareyJ79}. 

\begin{restatable}{theorem}{thmnphcourses}($\clubsuit$\footnote{Proof of all the theorems/lemmas/claims marked with $\clubsuit$ can be found in the supplementary.})
    \mefe  is \npc even when there are two courses with identical valuation functions; grades and valuations from TA side are $1$.
\label{thm:nph-2courses}\end{restatable}
Next, we show that the problem is \npc even when the degree of  courses and TAs is three, and capacity of each course is one. 
The result is due to a polynomial-time reduction from \tcm. 

\begin{restatable}{theorem}{thmnphcapone}($\clubsuit$)
    \mefe is \npc even when $v_j(t) = g_t(x_j)$ for all $t \in T, x_j \in X$, i.e., utility function is derived from grades and the capacity of each course is one; each TA has degree at most three, valuations and grades are in $\{0,1,2,3\}$, and no TA has same grade for any two courses. It is also \nph when each course has binary valuation function, and all the other constraints are same as earlier.
\label{thm:nph-cap1}\end{restatable}

Our next intractability result is when TAs have binary valuations and grades in the range $\{0,1,2\}$, other than the constraints on degrees and capacities. The result is due to a polynomial-time reduction from the \dmatching problem. 
\begin{restatable}{theorem}{nphcaptwo}($\clubsuit$)
    \mefe is \npc even when each course has capacity two, degree six,  and  valuations  in $\{0,1,2,3\}$; 
    each TA has binary valuation function, and grades in $\{0,1,2\}$. 
\label{thm:nph-cap2}\end{restatable}

\section{Polynomial Time Tractable Cases}

In this section, we identify some instances that can be solved in polynomial time. We first consider the case when the difference of the degree and capacity of each course is at most one.  

\begin{theorem}\label{degree-cap1}
   \mefe can be solved in polynomial time when the difference between the degree and capacity of each course is at most one.
\end{theorem}

\begin{proof}
    We visualize an instance ${\cal I}$ of \mefe as a bipartite graph $G=(X,T)$ and note that $G$ can be a disconnected graph. However, it is sufficient to solve each component separately, as
    we do not assign $0$-valued TAs to a course. Let $G'$ be a maximal connected subgraph of $G$. For the ease of notation, we reuse $X,T$ as the bipartiton of $G'$ as well, where $X, T$ is the set of courses and TAs, respectively, in $G'$. Let $X=\{x_1,\ldots,x_n\}$ and $T=\{t_1,\ldots,t_m\}$. We find an \mefem matching, if it exists, for $G'$. We design the algorithm based on the structure of the graph. We first consider the case the graph has a course with degree equalling capacity. If not, then we further consider other cases: whether it is  acyclic, or contains one cycle, or contains at least two cycles. In the case of two cycles, we argue that it is no-instance. In the other cases, for a yes-instance, the intuitive idea is that there exists a TA/course whose assignment or decision that it is unassigned dictates the whole matching. Note that such a TA/course should be chosen carefully as every choice does not fix the matching. 
    {\bf Case 1: There exists a course $x_i$ such that $d(x_i) - c_i = 0$.} Let $X'=\{x_1,\ldots,x_\ell\}$ be the set of courses in $X$ such that $d(x_i) = c_i$, where $i\in [\ell]$. We construct a matching $\mu$ as follows: for each $i\in [\ell]$, $\mu^{-1}(x_i)=N(x_i)$. If $\mu^{-1}(x_i)$ and $\mu^{-1}(x_j)$ are not disjoint, for any $i,j\in [\ell]$, then we return {\sf NO}. Otherwise, we call $\texttt{Extended Matching}(G',\mu,T_\mu=N(X'),X_\mu=X')$(Algorithm~\ref{alg:extend-match}) and find a matching $\mu$. The intuition is that we consider a course $x\in X$ whom we did not assign TAs yet, such that one of the positively valued TAs it has an edge with has already been assigned to a course. Thus, all the remaining TAs positively valued by this course must be assigned to it (we will argue in the correctness proof that for a yes-instance, only one of its TA is assigned to another course). If $\mu$ is a solution to ${\cal I}$, we return it; otherwise, we return {\sf NO}.

    \begin{algorithm}[]
\caption{{\texttt{Extended Matching}}}\label{alg:extend-match}
\textbf{Input:} a connected graph $G' = (X, T)$, a matching $\mu$, a set of TAs $T_\mu$, a set of courses $X_\mu$ \\
\textbf{Output:} a matching $\mu$
\begin{algorithmic}[1]
    \While{$\exists x \in N(T_{\mu})$ such that $x \notin X_{\mu}$}
        \For{all TA $t'\in T \setminus T_{\mu}$ s.t. $xt' \in E(G')$}
            \State $\mu(t') = x$
            \State $T_{\mu} \rightarrow T_{\mu} \cup \{t'\}$
        \EndFor 
        \State $X_{\mu} \rightarrow X_{\mu} \cup \{x\}$
    \EndWhile
\Return {$\mu$}
\end{algorithmic}
\end{algorithm}

{\bf Case 2: For all the courses $x_i\in X$, $d(x_i)-c_i=1$}

    {\bf Case 2.1: $G'$ is acyclic, i.e., a tree.}\label{c1:tree} In this case, $|E(G')|=|X|+|T|-1$. Furthermore, since $G'$ is a bipartite graph, $E(G')=\sum_{i\in [n]}d(x_i)$. Thus, $$\sum_{i\in [n]}d(x_i)=|X|+|T|-1$$  This implies that  $\sum_{i\in [n]}(d(x_i)-1)=|T|-1$
    Since $c_i=d(x_i)-1$, we have that  $$\sum_{i\in [n]}c_i=|T|-1$$ 
    Hence, there is only one TA that is unassigned in any feasible solution to ${\cal I}$. We guess this TA, say $t$. Now, we need to assign all TAs in $T'=T\setminus \{t\}$ to the courses in $X$. The intuitive idea is that since $t$ is unassigned, for all the courses that have positive valuations to $t$, we only have capacity many choices. Thus, the allocation of all these courses is fixed. Since the graph is connected, this eventually fixes the allocation of all the courses. Algorithm~\ref{alg:deg-cap-tree} describes the algorithm formally.  

\begin{algorithm}[]
\caption{Unit difference between degree and capacity: Acyclic Case}\label{alg:deg-cap-tree}
\textbf{Input:} a connected graph $G'=(X, T)$ \\
\textbf{Output:} either a matching $\mu$, or {\sf NO}.

\begin{algorithmic}[1]
\For{$t \in T$} \label{step:loopt}
    \State $T_{\mu_t} = {t}$, $X_{\mu_t} = \emptyset$, $\mu_t(z)=\emptyset$, for all $z\in T$
   \State $\mu_t$={\texttt{ Extended Matching}}($G',\mu_t,T_{\mu_t},X_{\mu_t}$)
    
    \If{$\mu_t$ is a solution to ${\cal I}$}   \Return $\mu_t$
    \EndIf
\EndFor
\Return {\sf NO}
\end{algorithmic}
\end{algorithm}

    {\bf Case 2.2: $G'$ contains only one cycle.} Let $C=(x_1,t_1,\ldots,x_\ell,t_\ell)$ be the cycle in $G'$. Clearly, there exists an edge in $C$ whose deletion makes the graph $G'$ acyclic. Thus, in this case, $|E(G')|=|X|+|T|$.  
    Furthermore, $|E(G')|=\sum_{i\in [n]}d(x_i)$.  Thus,  by equating the above two expressions just as we did in Case 2.1, we get $|T|=\sum_{i\in [n]}c_i$. Therefore, all TAs need to be assigned in a feasible matching for this case. Let $X_{C},T_C$ denote the set of courses and TAs, respectively, in $C$. We guess whether $t_1$ or $t_\ell$ is assigned to $x_1$ (one of them has to be assigned as $d(x_1)-c_1=1$). After this guess, an allocation for all the courses gets ``fixed'', i.e., there is a unique choice of allocation. Algorithm~\ref{alg:deg-cap-1cycle} describes the algorithm formally. 
\begin{algorithm}[]
\caption{Unit difference between degree and capacity: Unique cycle case}\label{alg:deg-cap-1cycle}
\textbf{Input:} a connected graph $G'=(X, T)$ \\
\textbf{Output:} either a matching $\mu$, or {\sf NO}.

\begin{algorithmic}[1]
\State let $C=(x_1,t_1,\ldots,x_\ell,t_\ell)$
\State let $\mu^{-1}(x_1)=N(x_1)\setminus \{t_{\ell}\}$ \label{stp:cycle-c0} 
\State $X_\mu=\{x_1\}$,  
$T_\mu=N(x_1)\setminus \{t_{\ell}\}$, 
\State $\mu$={\texttt{Extended Matching}}($G',\mu,T_\mu,X_\mu$) \label{stp:cycle-c1}
\If{$\mu$ is a solution to ${\cal I}$}
\Return $\mu$
\EndIf
\State $\tilde{\mu}^{-1}(x_1)=N(x_1)\setminus \{t_1\}$
\State $\tilde{\mu}^{-1}$={\texttt{Extended Matching}}($G',\tilde{\mu}^{-1},T_{\mu},X_{\mu}$)

    \If{$\tilde{\mu}^{-1}$ is a solution to ${\cal I}$}
\Return $\mu'$
    \EndIf
\Return {\sf NO}
\end{algorithmic}
\end{algorithm}

    {\bf Case 2.3: $G'$ contains more than one cycle.} The algorithm  returns {\sf NO} in this case.

The correctness can be found in supplementary. 
\end{proof}

Due to Theorem~\ref{degree-cap1}, we have the following result, which is in contrast to Theorem~\ref{thm:nph-cap1} when considering the degree of courses.

\begin{corollary}\label{cor:deg2}
    \mefe can be solved in polynomial time when the degree of courses is at most two.
\end{corollary}

Our next tractability result is in contrast to Theorem~\ref{thm:nph-cap1}, as it reduces the problem to \ssmtiffull (\ssmtif) in polynomial time, which can be solved in polynomial time~\cite{DBLP:conf/isaac/Kunysz18}

\begin{restatable}{theorem}{polycapone}($\clubsuit$)
     \mefe can be solved in polynomial time when the capacity of each course is one, and each TA has distinct valuations for positively valued courses.
\label{poly-cap1}\end{restatable}

Our next result is in contrast to Theorem~\ref{thm:nph-2courses}. The result is due to the observation that a TA with high grade needs to be prioritised over low grade TAs in a solution. Thus, we sort the TAs based on their grades and match the top capacity many TAs to the course. We also observe that an instance can be partitioned into distinct sub-instances, where each sub-instance has only one course.

\begin{restatable}{theorem}{polyTAdegone}($\clubsuit$)
       \mefe can be solved in polynomial time when the degree for each TA is one.
\label{thm:poly-degree1}\end{restatable}

\begin{corollary}\label{cor:deg1}
    \mefe can be solved in polynomial time when the number of courses is one.
\end{corollary}

\begin{proof}

The degree of every TA is 1 when the number of courses is one. Hence Theorem~\ref{thm:poly-degree1} is applicable.
\end{proof}

Another restricted instance we analyze is one in which the capacities of all courses are constant, and the  number of courses is also constant. In this case, the total number of possible solutions are polynomial which can be enumerated in polynomial time. Hence, we have the following result. 

\begin{restatable}{theorem}{polyconstantcoursecap}($\clubsuit$)
    \mefe can be solved in polynomial time when the number of courses and capacity of each course is constant.
\label{poly-constant-course-cap}\end{restatable}

\hide{
\begin{proof}
In this case, we argue that the total number of feasible matchings are at most $m^{\OO(1)}$, where $m$ is the number of TAs. Thus, we can try all possible feasible matchings and return the one that is a solution. If no matching meets our fairness criteria, then we return {\sf NO}. 
Note that the total number of TAs that get assigned in a matching is a constant value $c=\sum_{x_i \in X} c_i$. Given the set of TAs $T$, there are $\binom{m}{c}$ ways of choosing the set of matched TAs. For each of those ways, we need to further choose $c_i$ TAs for each course $x_i$. The number of such choices for each course $x_i$ are bounded by $\binom{c}{c_i}$, which is a constant given that $c$ is a constant. Since the total number of courses is constant, the product of all the possible choices for courses given a constant $c$ will also be a constant. Therefore, the total number of possible assignments is $\binom{m}{c} \cdot \OO(1)$ which is $m^{\OO(1)}$.
\end{proof}}

Next, we consider the case when every course has only two positive distinct values for TAs, and obtain the tractability, with some restrictions on valuations and grades of TAs. 

\begin{restatable}{theorem}{2val}
\mefe can be solved in polynomial time when every course has only two distinct positive valuations, and for a course no two TAs have the same grade. Furthermore, each TA has distinct valuations for positively valued courses.
\label{thm:2val}\end{restatable}

\begin{proof}

Let ${\cal I}=(X,T,\{v_i\}_{i\in X}, \{u_i\}_{i\in T}, \{g_i\}_{i\in T},$ $\{c_i\}_{i\in X},k\}$ be an instance of \mefe such that $v_i\colon T\rightarrow \{0,q_{i},q'_{i}\}$, where $q_i,q'_i$ are positive integers, and for any pair of TAs $t_i,t_j \in T$, $g_i(x)\neq g_j(x)$, for any course $x$. Without loss of generality, let $q_i \geq q'_{i}$. Next, we find the number of $q_i$ and $q'_i$ valued items assigned to $x_i$ in a solution. Let $a_{i}, a'_{i}$ be the number of $q_i$ and $q'_i$ valued items assigned to $x_i$, respectively.  If $q_i=q_i'$, then $a_i=c_{i}, a'_i=0$. Otherwise, we find the following equations for each course independently.  Due to our feasibility and satisfaction constraints for courses, we know that for each $i\in [n]$,

\begin{align}
  a_iq_i + a'_iq'_i &\geq kc_i \label{eq:1} \\
  a_i + a'_i &= c_i \label{eq:2}
\end{align}

If the set of the above two equations does not have a solution, then clearly, $\cal{I}$ is a no-instance of \mefe. Suppose that there is a solution of the set of Equation~\ref{eq:1} and \ref{eq:2} for every course. If we solve these two equations, we obtain that $a'_i \leq \frac{c_i(q_i-k)}{q_i-q'_i}$ and $a_i \geq \frac{c_i(k-q'_i)}{q_i-q'_i}$. Note that if $a_i'$ is negative, then $a_i>c_i$, which violates the feasibility of a matching. Thus, if $q_i<k$, it is a no-instance. So next we assume that $q_i\geq k$. We choose $a_i=\max\{0,\lceil \frac{c_i(k-q'_i)}{q_i-q'_i}\rceil\}$ and $a'_i = c_i - a_i$. We next give a polynomial time reduction from \mefe to the {\sc Stable Matching} problem for bipartite graphs, in which given a set of men, say $M$, and women, say $W$, a preference list of women over men, a preference list of men over woman, (preference lists might be incomplete); the goal is to decide whether there exists a matching such there is no pair of man and woman $(m,w)$ such that (i) $m$ is unmatched or prefers $w$ over his matched partner, and (ii) $w$ is unmatched or prefers $m$ over her matched partner. Such a pair of $(m,w)$ is called a \emph{blocking pair}, and a matching without blocking pairs is called a \emph{stable matching}.  

We create an instance of the {\sc Stable Matching} problem ${\cal J}$ as follows. The set of men is $M=T$, i.e., corresponding to every TA in $T$, we have a man in $M$. Corresponding to every course $x_i\in X$, we have $c_i$ women, say $w_i^1,\ldots, w_i^{c_i}$. If $a_i>0$, let $W_i=\{w_i^1,\ldots,w_i^{a_i}\}$, and if $a'_i >0$, let $W'_i=\{w_i^{a_i+1},\ldots,w_i^{c_i}\}$.  Next, we define the preference list of every woman $w$, say $P_w$, as follows. For every $w\in W_i$, man $t$ is in $P_w$ if and only if the valuation of course $x_i$ for the TA $t$ is $q_i$.  For every woman $w\in W'_i$, man $t$ is in the preference list of $w$ if and only if the valuation of course $x_i$ for TA $t$ is positive. Recall that in an instance the {\sc Stable Matching problem} a man $m$ is in the preference list of a woman $w$ if and only if $w$ is in the preference list of $m$. Next, we define the ordering of men in the preference list of every woman, which is based on the grades. Since no two TAs have the same grade for a course, for men $t,t'$ in $P_w$, woman $w$ prefers $t$ more than $t'$ if and only if $g_t(w)>g_{t'}(w)$. Next, we define the ordering of women in the preference list of every man. Let $P_t$ be the set of women in the preference list of the man $t$. Consider a man $t$. If $u_t(x_i) > u_t(x_j)$, where $x_i,x_j \in X$, then man $t$ prefers women in $w\in W_i \cup W'_i$ more than woman $\hat{w} \in W_j \cup W'_j$, where $w, \hat{w} \in P_m$. Furthermore,  if $w\in P_m\cap W_i$ and $\hat{w}\in P_m\cap W'_i$, then $m$ prefers $w$ more than $\hat{w}$. If $w,\hat{w} \in P_m \cap W_i$, then $m$ order them arbitrarily; similarly women in $P_m\cap W'_i$ are ordered arbitrarily. This completes the construction of ${\cal J}$. Proof of correctness can be found in supplementary. Since the {\sc Stable Matching} problem can be solved in polynomial time~\cite{GS62college}, the algorithm runs in polynomial time. 
\end{proof}

\section{Parameterized (Approximation) Algorithms}

In this section, we consider the problem in the realm of parameterized algorithms to cope with the intractability. We first consider the parameter $m$ (the number of TAs.). The following result is due to trying all possible partitions of TAs. 

\begin{restatable}{theorem}{FPTm}($\clubsuit$)
\mefe can be be solved in $\fpt(m)$, where $m$ is the number of TAs. 
\label{thm:FPTm}\end{restatable}

\hide{
\begin{proof}
Let ${\cal I}$ be a given instance of \mefe.
Note that every TA has $n$ options for the allocation. Thus, the total number of possible matching is $n^m$. For every possible matching, we can check in polynomial time if it is a solution to ${\cal I}$. Since, for a yes-instance, the number of courses is at least the number of TAs, this algorithm runs in $\OO(m^m(n+m)^{\OO(1)})$ time. 
 
\end{proof}
}

Next, we move our attention to the parameter $n$, the number of course. The next algorithm is similar to the one in Theorem~\ref{thm:2val}. But, here we have constant number of distinct positive valuations for each course, rather than just two. So, we guess the number of TAs  of each value assigned to a course in a solution. Note that we also consider constant capacity for each course. The rest of the algorithm is similar. In the algorithm in Theorem~\ref{thm:2val}, we know the upper bound for $a_i$, but here, we know the exact values, and hence add the edges accordingly in the {\sc Stable Matching} instance. 

\begin{restatable}{theorem}{FPTn}($\clubsuit$)
\mefe can be solved in $\fpt(n)$ when every course has constant number of distinct positive valuations and capacity, and for a course no two TAs have the same grade. Furthermore, each TA has distinct valuations for positively valued courses.
\label{thm:FPTn}\end{restatable}

Next, we design a parameterized approximation scheme with respect to the parameter $n,\epsilon$, and $\log v$, where $v$ is maximum value assigned to a TA by a course. Our basic idea is that for appropriately chosen $\epsilon'$, we guess the number of TAs assigned to course $x_i$ that have valuations in the range $[(1+\epsilon')^{j-1},(1+\epsilon')^j)$, for every course $x_i\in X$.  
Then, we create the number of copies of a course accordingly and reduce to the {\sc Stable Matching} problem  as in Theorem~\ref{thm:FPTn}. 

\begin{restatable}{theorem}{thmfptapxn}($\clubsuit$)
    \mefe admits an algorithm that given a yes-instance outputs a matching in which for every course $x$, ${\sf AvgUtil}(x)\geq (1-\epsilon)k$, where $0\leq \epsilon < 1$, when the capacity of each course is constant, no two TAs have the same grade for a course, and each TA assigns distinct utilities to each course, and runs in $\OO({\sf cap}^{\frac{n \log {\sf max_{val}}}{\log(1+\epsilon')}} (n+m)^{\OO(1)})$, where ${\sf cap}$ is the maximum capacity of a course and ${\sf max_{val}}$ is the maximum valuation of a course for a TA, $\epsilon'=\frac{1}{1-\epsilon}-1$. 
\label{thm:fpt-apxn}\end{restatable}

\section{Existence Results}\label{sec:existence}

In this section, we identify some yes and no instances of the problem.
\subsection{Yes-instances of \mefe}
We begin with identifying the instances  where the solution is guaranteed to exist. Before presenting our first result, we  define some notations.   

Let $r_i \colon T \rightarrow [m]$ be a rank function such that $t$ has $r_i(t)$th lowest grade in the course $x_i$, for each $i\in [n]$. Here, $m$ denotes the size of the set $T$.  
 Let $c = c_1 + \ldots + c_n$. We denote $k_i = g_{t}(x_i)$ where $t = r_i^{-1}(c)$. 
 Now, we are ready to present our first existential result.

\begin{restatable}{theorem}{binvalTAs}
($\clubsuit$)
   Let ${\cal I}$ be an instance of \mefe such that 
   
    (i) every TA $t_i \in T$ has the following valuation function for courses, given by $u_i\colon X \rightarrow \{0,a\}$, for each $i\in [m]$; 
    
    (ii) for every course $x\in X$ and TA $t\in T$, $v_x(t)=g_t(x)$, i.e., the utility of course $x$ for TA $t$ is the same as her grade in course $x$;

    (iii) for every subset $X'\subseteq X$, $|N(X')|\geq \sum_{x \in X'}^{} c_x$.

    (iv) for every course $x$, there should not be two distinct TAs $t, t'$ such that $u_{t}(x)=u_{t'}(x) = a$  but $g_{t}(x) = g_{t'}(x)$, i.e., no two TAs positively valuing a course should have the same grade for that course.
    
    Then, for
  $k = \min\{k_i \mid x_i \in X\}$, a \mefem matching always exists. 
    
\label{binvalTAs}\end{restatable}

Our next positive existence result is because   \mefe instance satisfying constraints in Theorem~\ref{thm:thmexist} can be reduced to \hrpfull (\hrp).   
\hide{
An instance I of \hrpfull involves a set of residents $R=\{ r_1,...,r_{n} \}$  and a set  of hospitals $H= \{h_1,...,h_m \}$. Each hospital $h_j \in H$ has a positive integral capacity $c_j$. Each hospital $h_j \in H$ has a preference list in which it ranks R in strict order similarly each resident $r_i \in R$ has a preference list in which they rank H in strict order. A hospital is under-subscribed in a matching $\eta$ if $| \eta(h_j) | < c_j$.A matching $\eta$ is valid if $\eta(h_j) \leq c_j$ for each $h_j \in H$ and $| \eta(r_i) | \leq 1$ for each $r_i \in R$. 
 A pair $(r_i,h_j)$ block a matching $\eta$ if $r_i$ is unassigned or prefers $h_j$ over $\eta(r_i)$ and $h_j$ under-subscribed or prefer $r_i$ over at least one resident in $\eta(h_i)$.
 A matching is said to be stable if it does not contain any blocking pair. In \hrp we are required to find a matching that is valid and stable. It is known that a valid stable matching in \hrp always exists~\cite{}.\pallavi{add citation.}
 }

\begin{restatable}{theorem}{thmexist}($\clubsuit$)
 Let ${\cal I}$ be an instance of \mefe such that all TAs and courses positively value each other and each TA must  assign distinct valuation to courses and all the TAs have  distinct grades/marks in a course. Then, \mefem always exist for $k=1$. 

\label{thm:thmexist}\end{restatable}

\subsection{No-instances of \mefe}

Each yes-instance discussed in the previous section has multiple constraints. In this section, we discuss that none of the constraints mentioned for positive existence results are by themselves sufficient for the existence of a \mefem matching. If course utilities are grades, i.e., $v_j(t) = g_t(x_j)$ for all $x_j \in X, t \in T$, with no other constraints, the answer to the existence question is NO. For example, in an instance with a single course having capacity 1 and 2 TAs having the same grade, the unassigned TA will always envy the other, so no \mefem matching exists. If $k=1$, with no other conditions, the answer is again NO, as the above example is a NO instance for this case. If all TAs and courses positively value each other, without further restrictions, the answer remains NO. For example, we take an instance with a single course having capacity 1 and 2 TAs, where both TAs positively value the course and the course positively values both TAs, and courses assign utilities less than $k$ to both TAs. For this instance, no \mefem matching exists, since the course satisfaction criteria cannot be met. If TAs assign valuation functions having the range ${0, a}$ to courses given that a feasible matching exists, the answer to the existence question remains NO as the previous example is a NO instance.

\section{Conclusion}
In this paper, we considered many-to-one matching problem under two-sided preferences, and inspired from the TA allocation problem to the courses, we used welfare function on one side and envy-freeness type fairness criteria on the other side. Indeed, the application of  problem is not limited to the TA allocation problem. We next list some of our specific open questions. Most of our algorithms require that ``in a course, no two TAs have the same grade''. Is this essential to derive those polynomial-time algorithms? In our $\fpt(n)$ algorithm, we require capacity and types of valuations of a course to be constant. What is the complexity with respect to $n$, when only one of them is constant?


\clearpage
\appendix
\begin{center}
    {\LARGE{\bf Appendix}}\footnote{Please ignore $\clubsuit$} \\
     \vspace{0.5cm}
     
\end{center}

\section{Other Applications of \mefe}

We discuss two practical applications of \mefe\footnote{We were unable to discuss these applications in the introduction due to lack of space.}. The \mefe model might be applied in a Job Market with a mediator company. Job applicants (analogous to TAs in our model) hold utilities over the firms (analogous to courses in our model) that the mediator company has ties with. Each of these firms might hold tests (grade in a course) specific to them, conducted by the mediator company. However, the firms assign utilities to applicants based on resume shortlisting, HR rounds along with the test scores. Here, applicants might envy each other based on their preferences (utilities assigned by TAs) and test scores for companies. And firms would want to ensure that they are satisfied with the applicants that are allocated to them. 
We could also use \mefe in Medical Residency Matching. Residency applicants (TAs) hold utilities over residencies (courses). Additionally, applicants might have grades in related areas to the residencies they are applying to. Residencies hold utility over applicants based on performance in relevant areas, clinical experience and interviews. 

\section{Missing proofs of Hardness Section}

\thmnphcourses*
\begin{proof}
    Let ${\cal S}=\{s_1,\ldots,s_m\}$ be an instance of the {\sc Equal-Cardinality Partition} problem. We construct an instance ${\cal J}$ of \mefe as follows. We have two courses, i.e., $X= \{x_1,x_2\}$, and $m$ TAs, $T=\{t_1,\ldots,t_m\}$. We define course valuations as follows: $v_i(t_j)=s_j$, where $i\in [2],j\in [m]$. The capacity of both the courses is $\frac{m}{2}$. All the TAs give $1$ as value to all the courses and they also have grade $1$ as in all the courses (values do not matter in our reduction, all the valuations (grades) of TAs need to be same). Let $k=\frac{\sum_{s\in {\cal S}}s}{2}$. Next, we prove the correctness. In particular, we prove the following.

\begin{restatable}{lemma}{lem:nph-2courses-correct}
        ${\cal S}$ is a yes-instance of {\sc Equal-Cardinality Partition} if and only if ${\cal J}$ is a yes-instance of \mefe.
\label{lem:nph-2courses-correct}\end{restatable}

\begin{proof}
    In the forward direction, let $S_1,S_2$ be a solution to ${\cal S}$. We construct a matching $\mu$ for ${\cal J}$ as follows: for $t\in S_i$, where $i\in [2]$, $\mu(t)=x_i$. Since $|S_1|=|S_2|=\frac{m}{2}$, $\mu$ is a feasible matching. Since $\sum_{s\in S_1}s=\sum_{s\in S_2}s = \frac{\sum_{s\in {\cal S}}s}{2}$, due to the construction of the instance ${\cal J}$ and the matching $\mu$, ${\sf AvgUtil}(x_i)=\frac{\sum_{s\in {\cal S}}s}{2}$, where $i\in [2]$. Since all the TAs are matched and everyone has same grade for both the courses, i.e., $0$, there is no envy between any pair of TAs. 

    In the reverse direction, let $\mu$ be a solution to ${\cal J}$. Let $S_i=\mu^{-1}(x_i)$. Due to the capacity constraint, every course is matched to $\frac{m}{2}$ TAs, thus, $|S_1|=|S_2|=\frac{m}{2}$. Since $k=\frac{\sum_{s\in {\cal S}}s}{2}$, $\sum_{s\in {\cal S}_1}s= \sum_{s\in {\cal S}_2}s = \frac{\sum_{s\in {\cal S}}s}{2}$. 
\end{proof}
This completes the proof.
\end{proof}

Next result is due to a polynomial-time reduction from \tcm, which is defined as follows. Given a set of $n$ men, say $M$, and $n$ women, say $W$, such that each of them specifies a preference list of size three that may contain ties; the goal is to decide whether there exists a perfect matching that is weakly stable. A matching $\eta$ is called weakly stable if there does not exist a pair $(m,w)$, where $m\in M$ and $w\in W$, each of whom is either unmatched in $\eta$ or prefer the other to their matched partner. \tcm is known to be \npc even when the ties belong to only women side~\cite[Theorem 5]{IRVING2009213}, and we consider this variant for our reduction. Suppose that for a woman $w$, the preference list is $\{m_1, m_2\} \succ m_3$. Then, we say that $w$ ranks $m_1,m_2$ at position $1$ and $m_3$ at position $2$. Furthermore, $m_1$ and $m_2$ are equally preferred by $w$.

\thmnphcapone*

\begin{proof}

Let ${\cal I}$ be an instance of \tcm such that the ties are only on the women's side. Let $M,W$ be the set of men and women, respectively. For $x\in M \uplus W$, let ${\sf pref}_x$ denote the preference list of $x$. We construct an instance ${\cal J}$ of \mefe as follows. Let $X=M$ and $T=W$, i.e., for each man $m \in M$, we have a course $m$, and for each woman $w$, we have a TA $w$. For every course $m\in X$, $v_m(w)=4-i$, if man $m$ ranks woman $w$ at position $i$; otherwise $0$. We next define a set of valuation functions of TAs. It is according to the preference list of women. If woman $w$ ranks man $m$ at position $i$, then the utility of TA $w$ for the course $m$ is $u_w(m)={4-i}$, if $m$ is not in the preference list of woman $w$, then $u_w(m)=0$. Next, we define a grade function for every course based the preference list of men, which is strict. If man $m$ ranks woman $w$ at position $i$, then the grade of TA $w$ for the course $m$ is $g_w(m)={4-i}$, if $m$ is not in the preference list of woman $w$, then $g_w(m)=0$. 
The capacity of each course is one and $k=1$.
Next, we prove the correctness. In particular, we prove the following.


\begin{figure}
\centering

\begin{subfigure}{0.4\textwidth}
\centering

\begin{tikzpicture}

\node at (1.7,4.2) {1};

\node at (0.3,4.2) {1};

\node at (0.3,3.65) {2};

\node at (1.7,3.4) {1};

\node at (0.3,2.8) {1};

\node at (1.7,2.8) {2};



\Vertex[x=2,y=4,size=0.2,label = $w_1   $,fontscale =1.5,position = right]{l11};
\Vertex[x=2,y=3,size=0.2,label = $w_2   $,fontscale =1.5,position = right]{l22};

\Vertex[x=0,y=4,size=0.2,label = $  m_1$,fontscale =1.5,position = left]{r11};
\Vertex[x=0,y=3,size=0.2,label = $ m_2$,fontscale =1.5,position = left]{r22};

\Edge(l11)(r11)
\Edge(l22)(r11)
\Edge(l22)(r22)

\end{tikzpicture}

\caption{An instance of \tcm}
\label{fig:tcm1}
\end{subfigure}

\vspace{5pt}

\begin{subfigure}[b]{0.4\textwidth}
\centering

\begin{tikzpicture}

\node at (1.7,4.2) {3};

\node at (0.3,4.2) {1};

\node at (0.3,3.65) {1};

\node at (1.7,3.4) {3};

\node at (0.3,2.8) {1};

\node at (1.7,2.8) {2};

\Vertex[x=0,y=4,size=0.2,label = $  c_1$,fontscale =1.5,position = left]{l1};
\Vertex[x=0,y=3,size=0.2,label = $  c_2$,fontscale =1.5,position = left]{l2};

\Vertex[x=2,y=4,size=0.2,label = $t_1 $,fontscale =1.5,position = right]{r1};
\Vertex[x=2,y=3,size=0.2,label = $t_2 $,fontscale =1.5,position = right]{r2};

\Edge[label = \color{blue}{3},fontsize=\large](l1)(r1)
\Edge[label = \color{blue}{2},fontsize=\large](r2)(l1)
\Edge[label = \color{blue}{3},fontsize=\large](l2)(r2)
\end{tikzpicture}

\caption{An instance of \mefe where grades are in the middle of the edges in \textcolor{blue}{blue} color}
\label{fig:mefe1}
\end{subfigure}

\caption{An instance of \mefe is presented in  \ref{fig:mefe1} corresponding to an instance of \tcm presented in \ref{fig:tcm1}}
\label{fig:thm2}
\end{figure}
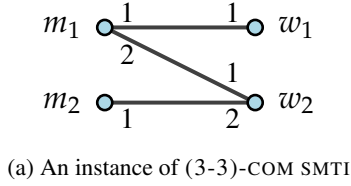
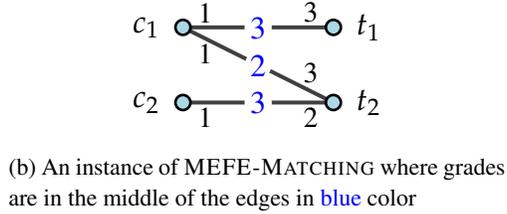

\begin{restatable}{lemma}{lem:nph-cap1-correct}
      ${\cal I}$ is a yes-instance of \tcm if and only if ${\cal J}$ is a yes-instance of \mefe.
\label{lem:nph-cap1-correct}\end{restatable}

\begin{proof}
  In the forward direction, let $\eta$ be a solution to ${\cal I}$. We claim that $\mu=\eta$ is also a solution to ${\cal J}$. Clearly, $\mu$ is a feasible matching as every course is matched to one TA. Furthermore,  for every course $w\in X$, ${\sf AvgUtil}(w) \geq 1$ as $v_m(w) \geq 1$ for all $w$ who are in preference list of $m$ and if a woman $w$ is matched to a man $m$ then $w$ must be in the preference list of $m$. Next, we argue that $\mu$ is merit-based envy-free. Note that $\mu$ is a perfect matching. Thus, every TA is matched to a course. Furthermore, every TA is matched to a course with non-zero utility due to the construction and the fact that $\eta$ is a solution to \tcm. Suppose that TA $t$ envies Ta $t'$, then we have $g_t(\mu(t')) \geq g_{t'}(\mu(t'))$. Due to construction and the fact that ties belong to only women side, equality is not possible. Hence, $g_t(\mu(t')) > g_{t'}(\mu(t'))$. Also, $u_t(\mu(t')) > u_t(\mu(t))$ which means $t \succ_{\eta(t')} t'$ and $\eta(t') \succ_{t} \eta(t)$, which contradict the fact that $\eta$ is a stable matching.

  In the reverse direction, let $\mu$ be a solution to ${\cal J}$. We claim that $\eta=\mu$ is a solution to ${\cal I}$. Recall that the capacity of every course is one in ${\cal I}$. Thus, $\eta$ saturates all men. Since, the number of courses and TAs is same and every course is matched to only one TA, every TA is matched to a course in $\mu$. Thus, $\eta$ is a perfect matching. Moreover, for every course $w\in X$, ${\sf AvgUtil}(w) \geq 1$ as we have $v_m(w) \geq 1$ for matched pair ($m,w$). Each man is matched with a woman who is in his preference list and hence, each woman also is matched with a man in her preference list due to construction. Next, we claim that $\mu$ is weakly stable. Towards the contradiction, suppose that there is a pair $(m,w)$ such that
  
  $w \succ_m \eta(m)$
  and 
 
  $m \succ_w \eta(w)$. Thus, due to the construction of the instance ${\cal J}$, $g_{w}(m)>g_{\eta(m)}(m)$, 
 
  and $u_w(m)>u_w(\mu(w))$. Thus, TA $w$ envies TA $w'$, which contradicts that $\mu$ is a solution to ${\cal J}$. 
\end{proof}

This completes the proof when course utilities are derived from grades. The binary case is similar. Here, the course valuation would be defined as follows: for every course $m \in X$, $v_m(w) = 1$ if woman $w$ is in the preference list of man $m$, otherwise $0$. The proof remains same for binary valuation function. 
\end{proof}

\nphcaptwo*

\begin{proof}
We give a polynomial-time reduction from the \dmatching (\dmatchingshort, in short) problem, in which 
 
  given three sets $P,Q,R$ of equal sizes and a set $E \subseteq P \times Q \times R$; 
  the goal is to find a matching $M \subseteq E$ such that for any two distinct triplets $(p_1, q_1, r_1), (p_2, q_2, r_2) \in M$, $p_1 \neq p_2, q_1 \neq q_2$, and $r_1 \neq r_2$, and every element of $P\uplus Q \uplus R$ belongs to a triplet in $M$ (such a matching is called a perfect matching). The \dmatchingshort problem is \npc even when each  element of $P\uplus Q \uplus R$ belongs to at most three sets in $E$~\cite{DBLP:books/fm/GareyJ79}.


\begin{figure}
    \centering

\begin{tikzpicture}
[scale=0.60]

\node[draw,align=right] at (2.5,5) {\color{blue}$u^X=2$\\ \color{blue}$g=1$};

\node[draw,align=right] at (-2.5,6.4) {\color{red}$u^X=3$\\ \color{red}$g=2$};

\node[draw,align=right] at (-3.2,-0.1) {$u^X=1$\\ $g=2$};

\draw (-0.2,1) ellipse (50pt and 127pt);
\Vertex[x=0,y=4,size=0.2,label = $r_{1}^{1}$,fontscale =1.5,position = above]{r11};
\Vertex[x=0,y=3,size=0.2,label = $r_{1}^{2}$,fontscale =1.5,position = left]{r12};
\Vertex[x=0,y=2,size=0.2,label = $r_{1}^{3}$,fontscale =1.5,position = left]{r13};
\draw [thick,dash dot] (-2,1.6) -- (1.5,1.6);

\draw [thick,dash dot] (-2,0.5) -- (1.5,0.5);

\Vertex[x=0,y=0,size=0.2,label = $r_{n}^{1}$,fontscale =1.5,position = left]{rn1};
\Vertex[x=0,y=-1,size=0.2,label = $r_{n}^{2}$,fontscale =1.5,position = left]{rn2};
\Vertex[x=0,y=-2,size=0.2,label = $r_{n}^{3}$,fontscale =1.5,position = left]{rn3};

\begin{scope}[shift={(-3,-3)}]

\draw (7.2,4.3) ellipse (50pt and 70pt);

\draw [thick,dash dot] (5.5,4.5) -- (9,4.5);

\Vertex[x=7,y=6,size=0.2,label = $p_{1}$,fontscale =1.5,position = right]{p1o};

\draw [dotted]  (7,6) -- (7,5);

\Vertex[x=7,y=5,size=0.2,label = $p_{n}$,fontscale =1.5,position = right]{pno};

\Vertex[x=7,y=4,size=0.2,label = $q_{1}$,fontscale =1.5,position = right]{q1o};

\draw [dotted]  (7,3) -- (7,4);

\Vertex[x=7,y=3,size=0.2,label = $q_{n}$,fontscale =1.5,position = right]{qno};
\end{scope}

\begin{scope}[shift={(-13.5,5)}]

\draw (7.2,-0.4) ellipse (50pt and 70pt);

\Vertex[label = $r_{1}^{d_{1}}$,x=8,y=1,size=0.2,fontscale =1.2,position = left]{R11d};
\Vertex[label = $r_{1}^{d_{2}}$,x=8,y=0,size=0.2,fontscale =1.2,position = left]{R12d};

\draw [dotted]  (8,0) -- (8,-1);

\Vertex[label = $r_{n}^{d_{1}}$,x=8,y=-1,size=0.2,fontscale =1.2,position = left]{Rn1d};
\Vertex[label = $r_{n}^{d_{2}}$,x=8,y=-2,size=0.2,fontscale =1.2,position = left]{Rn2d};
\end{scope}

\begin{scope}[shift={(-13.5,5)}]

\draw (7.2,-5.4) ellipse (50pt and 70pt);

\Vertex[label = $r_{1}^{d'_{1}}$,x=8,y=-4,size=0.2,fontscale =1.2,position = left]{R13d};
\Vertex[label = $r_{1}^{d'_{2}}$,x=8,y=-5,size=0.2,fontscale =1.2,position = left]{R14d};

\draw [dotted]  (8,-5) -- (8,-6);

\Vertex[label = $r_{n}^{d'_{1}}$,x=8,y=-6,size=0.2,fontscale =1.2,position = left]{Rn3d};
\Vertex[label = $r_{n}^{d'_{2}}$,x=8,y=-7,size=0.2,fontscale =1.2,position = left]{Rn4d};
\end{scope}

\Edge[color = blue](r11)(p1o)
\Edge[color = blue](r11)(q1o)

\Edge[color = red](r11)(R11d)
\Edge[color = red](r11)(R12d)

\Edge(r11)(R13d)
\Edge(r11)(R14d)

\end{tikzpicture}

    \caption{\mefem instance of \dmatching where $(p_1,q_1,r_1)$ is an edge. Edge for $r_n$ are not drawn}
    \label{fig:enter-label}
\end{figure}
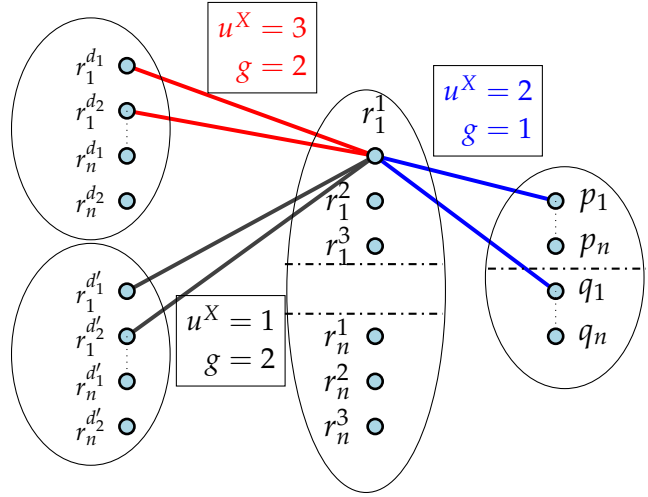

Let $(P,Q,R,E)$ be an instance of \dmatchingshort such that each  element of $P\uplus Q \uplus R$ belongs to at most three sets in $E$. 
We construct an instance of \mefe as follows. For each $r \in R$, we create three courses $r^{1},r^{2},r^{3}$. Intuitively, three copies of $r$ denote the sets containing $r$. 
For each $p\in P$, we add a TA $p$ and for each $q\in Q$, we add a TA $q$. We call these TAs as ``original'' and denote this set as $O$.  Next, we add some ``dummy'' TAs. For each $r \in R$, we add four dummy TAs $r^{d_1}$ , $r^{d_2}$ , $r^{d'_1}$ and $r^{d'_2}$. Let $D=\{r^{d_i}\colon r\in R, i\in [2]\}$ and $D'=\{r^{d'_i}\colon r\in R, i\in [2]\}$. Let $X$ be the set of all courses, i.e., $X=\{r^j \colon r\in R, j\in [3]\}$, and $T$ be the set of all the TAs, i.e., $T=O \uplus D \uplus D'$. Next, we define the utility function of every course. Suppose that $(p,q,r), (p',q',r)$, and $(p'',q'',r)$ be three triplets in $E$. Then, $v_{r^1}(p)=v_{r^1}(q)=2$, $v_{r^2}(p')=v_{r^2}(q')=2$, and $v_{r^3}(p'')=v_{r^3}(q'')=2$. Furthermore, for all $i\in [3]$, $v_{r^i}(r^{d_1})=v_{r^i}(r^{d_2})=3$, and $v_{r^i}(r^{d'_1})=v_{r^i}(r^{d'_2})=1$. That is, first copy of $r$ gives valuations $2$ to the TAs corresponding to the elements in the first set with $r$, second copy gives valuations $2$ to the TAs corresponding to the elements in the second set with $r$, and so on. All the copies gives valuation $3$ to the corresponding TAs in the set $D$ and $1$ to the corresponding elements in the set $D'$. The valuation for all the other TAs is $0$.

For every TA $t$ and course $x$, the utility function is as follows: $u_t(x)=1$, i.e., all the TAs value all the courses equally (the value does not matter in our reduction). 

Next, we define a grade function for each TA as follows.
For each TA $t$ corresponding to elements in $P\uplus Q$, $g_t(r_i^j)=1$, if $t$ and $r$ belong to a triplet in $E$. If $t$ is a dummy TA such that $t\in \{r^{d_1},r^{d_2},r^{d'_1},r^{d'_2}\}$, then, $g_t(r_i^j)=2$, where $i\in [n], j\in [3]$, and $0$ for all the other courses.

We set $k=2$ and the capacity for each course is $2$.

Next, we prove the correctness. In particular, we prove the following. 

\begin{restatable}{lemma}{lem:correctnees-nph-3DPM}
${\cal I}=(P,Q,R,E)$ is a yes-instance of \dmatchingshort if and only if ${\cal J}=(X,T,\{v_i\}_{i\in X}, \{u_i\}_{i\in T}, \{g_i\}_{i\in T}, \{c_i\}_{i\in X},k\}$ is a yes-instance of \mefe. 
\label{lem:correctnees-nph-3DPM}\end{restatable}

\begin{proof}
    In the forward direction, let $M$ be a solution to ${\cal I}$. We construct a solution to ${\cal J}$ as follows. If $(r,p,q)\in M$, then $\mu(p)=\mu(q)=r^1$. Furthermore, $\mu(r^{d_1})=\mu(r^{d'_1})=r^2$ and $\mu(r^{d_2})=\mu(r^{d'_2})=r^2$. We first show that $\mu$ is a feasible matching.
    \begin{restatable}{claim}{clm:matching-feasibility}
            $\mu$ is a feasible matching.
   \label{clm:matching-feasibility}\end{restatable}
    \begin{proof}
        Since $M$ is a perfect matching, for every $r\in R$, its first copy $r^1$ is matched to two TAs. By the construction, all the other courses are matched to two dummy TAs. Thus, $\mu$ is a feasible matching.
    \end{proof}
    Next, we argue that $\mu$ meets satisfaction criteria of each course.
    \begin{restatable}{claim}{clm:satisfaction}
               For each course $r^j\in X$, where $r\in R$, ${\sf AvgUtil}(r^j)\geq 2$.
\label{clm:satisfaction}\end{restatable}
    \begin{proof}
        Recall that every $r\in R$ is in a triplet in $M$. Suppose $(p,q,r)\in M$. Then, $\mu^{-1}(r^1)=\{p,q\}$. Since the utility of $r^1$ for $p$ and $q$ is $2$, ${\sf AvgUtil}(r^1)=2$. Since $r^2$ is matched with $r^{d_1}$ and $r^{d'_1}$, and its utility for $r^{d_1}$ is $3$ and $r^{d'_1}$ is $1$, ${\sf AvgUtil}(r^2)=2$. Similarly, since $r^3$ is matched with $r^{d_2}$ and $r^{d'_2}$, ${\sf AvgUtil}(r^3)=2$.
    \end{proof}
    Next, we prove merit-based envy-freeness between TAs. Note that since every TA values all the courses equally, if s/he is matched to a course, then there is no envy. Since $M$ is a perfect matching, every element in $P\uplus Q$ is in a triplet in $M$. Thus, every TA corresponding to elements in $P\uplus Q$ is matched to a course. TAs $r^{d_1} \in D, r^{d'_1} \in D'$ are matched to $r^2$, and $r^{d_2} \in D, r^{d'_2} \in D'$ are matched to $r^3$. Thus, all the TAs are matched to a course. This completes the proof in the forward direction.

In the reverse direction, let $\mu$ be a solution to ${\cal J}$. We begin with the following observations about the solution $\mu$. 

\begin{restatable}{observation}{obs1:rev-proof}
   Every course $x$ is matched to two TAs for whom the utility is non-zero. 
\label{obs1:rev-proof}\end{restatable}

\begin{proof}
    Since the capacity of every course is two, every course is matched to two TAs. Recall that $k=2$ and maximum utility of a course for any TA is $3$. Thus, $x$ is matched to TAs with non-zero utility. 
\end{proof}

\begin{restatable}{observation}{obs2:rev-proof}
    Every TA $t\in D\uplus D'$ is matched to a course in $\mu$. 
\label{obs2:rev-proof}\end{restatable}

\begin{proof}
    Suppose that a TA $t\in D\uplus D'$ is unsaturated in $\mu$. Recall that grade of $t$ is $2$ in every course, and no other TA has higher grade. Since $t$ is unsaturated, $t$ envies all the TAs by the definition of merit-based envy-freeness. 
\end{proof}

\begin{restatable}{observation}{obs3:rev-proof}
     If a TA $t \in D'$ matched to a course $x$, then $x$ is also matched to a TA in $D$.  
\label{obs3:rev-proof}\end{restatable}

\begin{proof}
    Recall that the capacity of every course is $2$ and $k=2$. Thus, course $x$ is matched to two TAs. Since the utility of a course for TAs in $D'$ is $1$, $x$ is matched with a TA in $D$ as the utility for TAs in $O$ is $2$ and for TAs in $D$ is $3$. 
\end{proof}
Due to Observations~\ref{obs2:rev-proof} and \ref{obs3:rev-proof}, without loss of generality, let $r^{d_1}, r^{d'_1}$ are matched to $r^{j}$ and $r^{d_2}, r^{d'_2}$ are matched to $r^{j'}$, where $r\in R, j,j' \in [3], j\neq j'$. Thus, $r^i$, where $i\in [3]\setminus \{j,j'\}$, is matched to two TAs in $O$, say $p,q$. Due to Observation~\ref{obs1:rev-proof} and the fact that $r^i$ has non-zero utility to $p,q \in O$, if $(p,q,r)\in E$, we know that if a course $x$ is matched to two TAs in $O$, then $r$ and $\mu^{-1}(r)$ forms a triplet in $T$.  

We construct a set $M\subseteq T$ as follows: $M=\{(r,p,q)\in T \colon \mu^{-1}(r^i)=\{p,q\}, \text{ where } i\in [3]\}$. Next, we argue that $M$ is a matching. Note that $p$ and $q$ are only matched to one course in $\mu$. Since the capacity of every course is also $2$, a course is matched to only two TAs. Thus, $M$ is a matching. Next, we argue that $M$ is a perfect matching. Recall that $k=2$ and the capacity of every course is also $2$, due to Observation~\ref{obs3:rev-proof}, only two courses corresponding to $r\in R$ can be matched to TAs in $D\uplus D'$. Morevover, since capacity of each course is two, one copy is matched to two TAs in $O$. Hence, every $r\in R$ belongs to a triplet in $M$. Since $|P|=|Q|=|R|$ and $M$ is a matching, every element of $P\uplus Q$ also belong to a triplet in $M$. 
\end{proof}
This completes the proof.
\end{proof}

\section{Missing proof of Polynomial Time Tractable Cases}

\subsection{Correctness of Theorem~\ref{degree-cap1}}

\begin{proof}
    
Here, we prove the correctness of the algorithm in Theorem~\ref{degree-cap1}. In particular, we show the correctness of each case individually. 

\begin{restatable}{lemma}{lem:degree-capC1}
    If there exists a course $x_i\in X$ such that $d(x_i)-c_i=0$, then ${\cal I}$ is a yes-instance of \mefe if and only if we return  a matching in Case 1.
\label{lem:degree-capC1}\end{restatable}
\begin{proof}
To prove the forward direction, first let us assume that $\eta$ is a solution to ${\cal I}$. Recall that $X'=\{x_1,\ldots,x_\ell\}$ is the set of courses for which  $d(x_i) = c_i$. Thus,  for all the courses $x_i \in X$, $\mu^{-1}(x_i)=\eta^{-1}(x_i)=N(x_i)$. 

Due to the connectivity of the graph $G'$, and the fact that if we assign TAs to a course in the algorithm, it does not have any unallocated neighbors, allocated TAs have unallocated courses at every step of the algorithm. 
Let $(x_{\ell + 1}, \ldots, x_n)$  be the sequence of courses considered by this algorithm after we call $\texttt{Extended Matching}$ given that $X_\mu = X'$ and $T_\mu = N(X')$. We next prove the following.

\begin{restatable}{claim}{clm:d-c=0}
    For every $i\in \{\ell +1,\ldots,n\}$, $\mu^{-1}(x_i)=\eta^{-1}(x_i)$.
\label{clm:d-c=0}\end{restatable}

\begin{proof}  
We prove it by strong induction on $i$.

    \emph{Base Case.} $i=\ell+1$. Since $x_{\ell + 1}$ is the first course considered after the assignments of courses in $X'$, there must be a TA $t' \in T_\mu$ such that $x_{\ell + 1}t' \in E(G')$. Furthermore, $\eta^{-1}(x_{\ell + 1}) \cap T_\mu = \emptyset$ as $\mu^{-1}(x_i)=\eta^{-1}(x_i)$, for all $i\in [\ell]$. This indicates that $t'$ is the only neighbor of $x_{\ell + 1}$ that is in $T_\mu$ at this step and shows that $\eta^{-1}(x_{\ell + 1}) = N(x_{\ell + 1}) \setminus \{t'\}$. Also, by the $\texttt{Extended Matching}$ procedure, for all the TAs $t'' \in N(x_{\ell + 1})\setminus \{t'\}$, $\mu(t'')=x_{\ell + 1}$. Hence, $\mu(x_{\ell + 1})=\eta(x_{\ell + 1})$. 

    \emph{Induction Step.} Suppose the claim is true for all $i\leq j-1$. We next prove it for $i=j$. We considered $x_j$ in the algorithm as $x_j\in N(T(\mu))$. Let $x_jt' \in E(G')$, where $\mu(t')=x_i$, $i<j$. Due to induction hypothesis, $\eta(t')=x_i$. Furthermore, due to induction hypothesis, $\eta(x_i)=\mu(x_i)$, for all $i\leq j-1$. Thus,  $\eta(x_j)\cap T(\mu) = \emptyset$, where $T(\mu)$ is the set of TAs constructed in the first $j - \ell - 1$ iterations of the while loop of Algorithm~\ref{alg:extend-match} along with the TAs in $N(X')$, i.e., it contains the set of TAs assigned to $x_1,\ldots,x_{j-1}$. Thus, $t'$ is the only neighbor of $x_j$ that is in  $T(\mu)$ at this step. Thus, for all the TAs $t'' \in N(x_j)\setminus \{t'\}$  $\mu(t'')=x_j$. Hence, $\mu(x_j)=\eta(x_j)$. 
\end{proof}  

In the reverse direction, if we return a matching in Case 1, then ${\cal J}$ is clearly a yes-instance of \mefe.
\end{proof}
\begin{restatable}{lemma}{lem:degree-capC1}
    If $G'$ is acyclic, then ${\cal I}$ is a yes-instance of \mefe if and only if Algorithm~\ref{alg:deg-cap-tree} returns a matching.
\label{lem:degree-capC1}\end{restatable}
\begin{proof}

In the forward direction, let $\eta$ be a solution to ${\cal I}$. As argued in Case 1, there exists a TA $\tilde{t}$, that is unallocated in $\eta$. Consider the case when $t=\tilde{t}$ in Step~\ref{step:loopt} of Algorithm~\ref{alg:deg-cap-tree}. We claim that $\eta=\mu_{\tilde{t}}$.  Let $x_1,\ldots, x_n$ be the sequence of courses that are considered in Algorithm~\ref{alg:extend-match} when called for $T_\mu=t$. We prove it by strong induction on $i\in [n]$. In particular, we prove the following. 

\begin{restatable}{claim}{clm:tree}
    For every $i\in [n]$, $\mu^{-1}(x_i)=\eta^{-1}(x_i)$.
\label{clm:tree}\end{restatable}

\begin{proof}
Recall that at any step of the algorithm, $T_{\mu_{\tilde{t}}}, X_{\mu_{\tilde{t}}}$ are the sets of matched TAs and courses, respectively, so far.  
    \emph{Base Case.} $i=1$. In the first iteration $T_{\mu_{\tilde{t}}}=\{\tilde{t}\}$.  Thus, $x_1\in N(\tilde{t})$. Since $\tilde{t}$ is unmatched in both $\eta$ and $\mu$ and the $d(x_1)-c_1 = 1$, $\mu(x_1)=\eta(x_1)$. 

    \emph{Induction Step.} Suppose the claim is true for all $i\leq j-1$. We next prove it for $i=j$. We considered $x_j$ in the algorithm as $x_j\in N(T(\mu_{\tilde{t}}))$. Let $x_jt' \in E(G')$, where $\mu_{\tilde{t}}(t')=x_i$, $i<j$. Due to induction hypothesis, $\eta(t')=x_i$. Furthermore, due to induction hypothesis, $\eta(x_j)\cap T(\mu_{\tilde{t}}) = \emptyset$, where $T(\mu_{\tilde{t}})$ is the set of TAs constructed in the first $j-1$ iterations of the while loop of Algorithm~\ref{alg:extend-match}, i.e., it contains $\tilde{t}$ and the set of TAs assigned to $x_1,\ldots,x_j$. Thus,  $t'$ is the only neighbor of $x_j$ that is in  $T(\mu_{\tilde{t}})$ at this step. Thus, for all the TAs $t'' \in N(x_j)\setminus \{t'\}$  $\mu(t'')=x_j$. Hence, $\mu(x_j)=\eta(x_j)$.  
\end{proof}

In the reverse direction, if Algorithm~\ref{alg:deg-cap-tree} returns a matching, then clearly ${\cal J}$ is a yes-instance.

\end{proof}

The idea of the next proof is similar to Lemma~\ref{lem:degree-capC2}
\begin{restatable}{lemma}{lem:degree-capC2}
    If $G'$ contains only one cycle, then ${\cal I}$ is a yes-instance of \mefe if and only if Algorithm~\ref{alg:deg-cap-1cycle} returns a matching.
\label{lem:degree-capC2}\end{restatable}
\begin{proof}
    In the forward direction let $\eta$ be a solution to ${\cal I}$. Let $C$ 
    be the unique cycle in $G'$. Let $x_1 \in C$. Let us also say that $x_1$ shares $t_1$ with $x_2 \in C$ and $t_\ell$ with $x_l \in C$. Since $d(x_1)-c_1=1$, either $t_1\in \eta^{-1}(x_1)$ or $t_\ell \in \eta^{-1}(x_1)$. Note that the number of courses in $C$ is equal to the number of TAs. This means that we cannot assign both $t_\ell$ and $t_1$ to $x_1$, otherwise, there would be one course $x_i \in C$ which gets neither of TAs that belong to $C$, and $x_i$ has at most $(d(x_i)-2)$ TAs in $T\setminus C$. Suppose that $t_1\in \eta^{-1}(x_1)$. Then, $\eta(t_\ell) \neq x_1$. In order to satisfy feasibility for $x_1$, $\eta^{-1}(x_1)=N(x_1)\setminus \{t_\ell\}$. If we assign $t_\ell$ to $x_1$ instead of $t_1$, then $\eta^{-1}(x_1)=N(x_1)\setminus {t_1}$. We first argue for the case when $\eta^{-1}(x_1)=N(x_1)\setminus \{t_\ell\}$. The other case can be argued analogously. We consider the case when  $\mu^{-1}(x_1)=N(x_1)\setminus \{t_\ell\}$ in Step~\ref{stp:cycle-c0} of Algorithm~\ref{alg:deg-cap-1cycle}. Thus, $\mu^{-1}(x_1)=\eta^{-1}(x_1)$. 
We argue that the matching $\mu$ returned in Step~\ref{stp:cycle-c1} of Algorithm~\ref{alg:deg-cap-1cycle} is same as $\eta$. 
Note that since the graph is connected, as argued earlier,  Algorithm~\ref{alg:extend-match} terminates when all the courses are matched. We already argued that $\mu^{-1}(x_1) = \eta^{-1}(x_1)$.  For the remaining courses, we prove the result similar to Claim~\ref{clm:tree}. Let $(x_2,\ldots, x_n)$ be the sequence of courses that are considered in Algorithm~\ref{alg:extend-match} when called for $X_\mu=\{x_1\}$ and $T_\mu=\mu^{-1}(x_1)$ (we reused the notation $x_2$ for convenience, it could be different from the one in $C$ that shares the TA with $x_1$). We prove it by strong induction on $i\in [n]$. In particular, we prove the following. 

\begin{restatable}{claim}{clm:1cycle}
    For every $i\in \{2,\ldots,n\}$, $\mu^{-1}(x_i)=\eta^{-1}(x_i)$.
\label{clm:1cycle}\end{restatable}

\begin{proof}
  
    \emph{Base Case.} $i=2$.  Since we first considered $x_2$ in Algorithm~\ref{alg:extend-match}, it has a neighbor in $N(x_1)\setminus \{x_\ell\}$, say $t'$. Since $\eta(x_1)=\mu(x_1)$ and $\eta$ is a solution, $t'$ can be the only TA in $N(x_2)$ that belongs to $T_\mu$, else $x_2$'s feasibility will not be met. Thus, for every $t''\in N(x_2)\setminus \{t'\}$, $\mu(t'')=x_2$. Thus, $\eta(x_2)=\mu(x_2)$.

    \emph{Induction Step.} Suppose the claim is true for all $i\leq j-1$. We next prove it for $i=j$. We considered $x_j$ in the algorithm as $x_j\in N(T(\mu))$. Let $x_jt' \in E(G')$, where $\mu(t')=x_i$, $i<j$. Due to induction hypothesis, $\eta(t')=x_i$. Furthermore, due to induction hypothesis, $\eta(x_j)\cap T(\mu) = \emptyset$, where $T(\mu)$ is the set of TAs constructed in the first $j - 2$ iterations of the while loop of Algorithm~\ref{alg:extend-match} along with the TAs assigned to $\{x_1\}$, i.e., it contains the set of TAs assigned to $x_1,\ldots,x_{j-1}$. Thus, $t'$ is the only neighbor of $x_j$ that is in  $T(\mu)$ at this step. Thus, for all the TAs $t'' \in N(x_j)\setminus \{t'\}$  $\mu(t'')=x_j$. Hence, $\mu(x_j)=\eta(x_j)$. 

    The inductive proof for the case where $t_\ell$ is assigned to $x_1$ initially is analogous to the proof given above.
\end{proof}  
In the reverse direction, if Algorithm~\ref{alg:deg-cap-1cycle} returns a matching, then clearly ${\cal J}$ is a yes-instance.
\end{proof}

\begin{restatable}{lemma}{lem:degree-capC3}
    If $G'$ contains more than one cycle, then ${\cal I}$ is a no-instance of \mefe.
\label{lem:degree-capC3}\end{restatable}
\begin{proof}
    The number of edges in $G'$ is  $|E(G')|=\sum_{i\in [n]}d(x_i)$. We also know that in any $G'$ that is not acyclic, the number of edges is at least the number of vertices, which in this case is $|X| + |T|$. 
    
\begin{restatable}{claim}{clm:mulcycle}
In a connected graph $H$, if $|E(H)|=|V(H)|$, then $H$ contains only one cycle. 
\label{clm:mulcycle}\end{restatable}

\begin{proof}
    Since $H$ is connected and  $|E(H)|=|V(H)|$, it contains at least one cycle, say $C$. Let $e=uv$ be an edge in  $C$. Since the deletion of an edge that belongs to a cycle does not disconnect the graph, $H-e$ is connected. Furthermore, since the number of edges in $H-e$ is $|E(H)|-1$, $H-e$ is a tree. Thus, there is a unique path between $u$ and $v$ in $H-e$. Suppose that there exist two distinct cycles $C,C'$ in $H$ containing $e$. Then, $C-e$ and $C'-e$ are two two distinct paths between $u$ and $v$ in $H-e$, a contradiction.  
\end{proof}

Thus, due to Claim~\ref{clm:mulcycle}, $|E(G')|>|X|+|T|$. Thus,

$\sum_{i\in [n]}d(x_i)>|X|+|T|$. Recall that $d(x_i) - 1 = c_i$. Thus,  
$\sum_{i\in [n]}c_i > |T|$. Since the number of TAs is lesser than the number of TAs needed to meet each course's capacity constraint, it is impossible to generate a feasible matching. Therefore, ${\cal J}$ is a no-instance of \mefe in this case.
\end{proof}
Therefore, we have proved that the above algorithm finds an \mefe, if it exists, for any instance by considering all cases, showing the correctness of the algorithm.
\end{proof}

\polycapone*
\begin{proof}
Toward designing the algorithm, we reduce the problem to \ssmtiffull (\ssmtif), which is defined as follows. Given a set of men, say $M$, and women, say $W$, a preference list of women for each man, a preference list of men for each woman, (preference lists might be incomplete and/or contain ties), and weight for every acceptable pair; the goal is to decide whether there exists a maximum weight strongly stable matching. A pair $(m,w)$ is called \emph{acceptable} if they are in each others' preference lists. A matching $\eta$ is called strongly stable if there does not exist an unmatched pair $(x,y)$ such that (i) either $x$ is unmatched in $\eta$ and $y$ is in the preference list of $x$, or $x$ strictly prefers $y$ over his/her matched partner in $\eta$, and (ii) either $y$ is unmatched in $\eta$ and $x$ is in the preference list of $y$, or $y$ is indifferent between $x$ and his/her matched partner, or $y$  strictly prefers $x$ over his/her matched partner in $\eta$. 
A maximum weight strongly stable matching can be found in polynomial time~\cite{DBLP:conf/isaac/Kunysz18}. 

Given an instance ${\cal I}$ of \mefe satisfying the constraints of Theorem 5, we create an instance ${\cal J}$ of \ssmtif as follows. The set of men $M$ is the same as the set of courses, and the set of women $W$ is the same as the set of TAs. 

If $u_t(x) > 0$, then only man $x$ and woman $t$ are in each others' preference list.
We set the preferences of men according to the grades of TAs. For two TAs $t,t'$ and a course $x$ such that $u_t(x) > 0$ and $u_{t'}(x) > 0$, if $g_t(x)>g_{t'}(x)$, then man $x$ strictly prefers $t$ over $t'$, and if $g_t(x)=g_{t'}(x)$, then $x$ prefers $t$ and $t'$ equally. Next, we define preference lists of women which is on the basis of utilities of TAs that are strict. For a TA $t\in T$ such that $u_t(x) > u_t(y)$, woman $t$ prefers man $x$ over $y$, i.e., $x \succ_t y$. 
For an acceptable pair $(x,t)$ such that  $v_x(t) < k$, the weight of $(x,t)$ is $0$. The weight of all other acceptable pairs is $1$. 

Let ${\cal J}$ be this constructed instance of \ssmtif. Note that the women's preference list do not have ties. We find a solution $\eta$ to ${\cal J}$ using the known polynomial-time algorithm~\cite{DBLP:conf/isaac/Kunysz18}. 
If the algorithm returns a strongly stable matching with weight $n$ (i.e., the number of men), then we return {\sf YES}; otherwise {\sf NO}. 
In particular, we prove the following.

\begin{restatable}{lemma}{polyssmti}
   ${\cal I}$ is a yes-instance of \mefe if and only ${\cal J}$ has a strongly stable matching with weight $n$. 
\label{lem:poly-ssmti}\end{restatable}

\begin{proof}
    In the forward direction, let $\mu$ be a solution to ${\cal I}$. We claim that $\eta=\mu$ is a solution to ${\cal J}$. Since the capacity of every course is at least one, due to the feasibility condition, all the courses are matched in $\mu$. Thus, all the men are matched in $\eta$. Furthermore, since the ${\sf AvgUtil}(x)\geq k$ for every course $x$ with respect to the matching $\mu$, the weight of every edge in $\eta$ is $1$. Thus, the weight of matching $\eta$ is $n$.   We next argue that it is strongly stable.   Consider an unmatched pair $(m,w)$. 
    Consider two TAs $w$ and $\eta(m)$. Since $\mu$ does not have merit-based envy, either $g_w(m)<g_{\mu(m)}(m)$ or $u_w(m)\leq u_{w}(\mu(w))$. In the former case, $m$ strictly prefers $\eta(m)=\mu(m)$ over $w$. In the latter case, since the preference list of $w$ is strict, she strictly prefers $\eta(w)=\mu(w)$ over $m$. 
    Therefore, $\eta$ is strongly stable and its weight is $n$. 

   In the reverse direction, let $\eta$ be a solution to ${\cal J}$ with weight $n$. We claim that $\mu=\eta$ is a solution to ${\cal I}$. Since the weight of $\eta$ is $n$, the maximum weight of any edge is $1$, and the number of men is $n$, $\eta$ saturates all men. Furthermore, $\eta$ does not contain any edge with weight $0$. Since the set of men is corresponding to the courses, all the courses are saturated in $\mu$. Clearly, only one TA is matched to every course in $\mu$. . 
   Since $\eta$ does not contain any edge with weight $0$, due to the construction of the instance ${\cal J}$,
   
  For any course $x$, $u_x^X(\mu^{-1}(x)) \geq k$, i.e., all courses are satisfied with their assignment. We show that $\mu$ is merit-based envy-free with respect to ${\cal I}$. Suppose not, then there exists TAs $t,t'$ such that $g_t(\mu(t')) \geq g_{t'}(\mu(t'))$ and $u_t(\mu(t)) < u_t(\mu(t'))$. 
  This implies that man $\eta(t')$ is either indifferent between $t$ and $t'$ or strictly prefers $t$ over $t'$. Furthermore, woman $t$ is either unmatched or prefers $\eta(t')$ over $\eta(t)$. Hence, due to the pair $(\eta(t'),t)$, $\eta$ is not strongly-stable, a contradiction.  Hence, $\mu$ is a solution to ${\cal I}$.
  
\end{proof}
This completes the proof.
\end{proof}

\polyTAdegone*
\begin{proof}
Let ${\cal I}$ be an instance of \mefe, where 
the degree of each TA is one. 
We create a set of sub-instances $I = \{{\cal I}_1, {\cal I}_2,\ldots, {\cal I}_n\}$ as follows. For every $j\in [n]$, ${\cal I}_j$ contains the course $x_j$ and its neighbors as the set of TAs. The utilities, grades and $k$ are the same as in ${\cal I}$. Let $x_j$ be the unique course in ${\cal I}_j$ and $c_j$ be its capacity. For any $x_j \in X$, if the number of TAs which are neighbours of $x_j$ are less than $c_j$, then clearly, ${\cal I}_j$ is a no-instance. If the number of TAs is $c_j$, then we match all the TAs to $x_j$ and check whether this assignment is a solution to ${\cal I}_j$ (i.e., it satisfies all constraints). If it is, then we have a trivial yes-instance for ${\cal I}_j$, otherwise we have a no-instance. 
Next, we consider the case when the number of TAs is more than $c_j$.
Our algorithm is based on the following idea: if a TA $t$ is unmatched, then he will envy a matched TA $t'$, if the grade of $t$ in $x_j$ is at least the grade of $t'$ in $x_j$. Thus, our algorithm proceeds as follows: sort the TAs based on their grades in $x_j$ and match the first $c_j$ many TAs to $x_j$. Let $\mu_j$ be the resultant matching. For any $x_j \in X$, if $\mu_j$ is not a solution to ${\cal I}_j$ (i.e., it does not satisfy all the constraints), then return {\sf NO}; otherwise, we construct a matching $\mu$ for ${\cal I}$ as follows: for every $j\in [n]$, and TA $t$ in the instance ${\cal I}_j$, $\mu(t)=\mu_j(t)$. If $\mu$ is a solution to ${\cal I}$, then we return $\mu$; otherwise, we return {\sf NO}. Next, we prove the correctness of the algorithm to generate a \mefem matching for ${\cal I}$. In particular, we prove the following two statements. 

\begin{restatable}{lemma}{course1P}
   If ${\cal I}_j$ is a yes-instance of \mefe, then $\mu_j$ is the unique solution to ${\cal I}_j$.  
\label{lem:course1P}\end{restatable}

\begin{proof}
    Let $\eta_j$ be a solution to ${\cal I}_j$. We claim that $\eta_j=\mu_j$. Suppose not, then there exists a TA $t$ that is matched in $\mu_j$, but not in $\eta_j$, Since $c_j$ TAs are matched in both $\mu_j$ and $\eta_j$, there is also a TA $t'$ that is matched in $\eta_j$, but not in $\mu_j$. Due to the construction of $\mu_j$, $g_t(x_j)\geq g_{t'}(x_j)$. Since $t$ is unmatched in $\eta_j$, $t$ envies $t'$ in $\eta_J$, a contradiction. 
\end{proof}

\begin{restatable}{lemma}{deg1P}
${\cal I}$ is a yes-instance of \mefe if and only if the above algorithm returns a matching.
   
\label{lem:deg1P}\end{restatable}

\begin{proof}
In the forward direction, let $\eta$ be a solution to ${\cal I}$. If we denote the matching corresponding to each course, $x_j \in X$, by $\eta_j$, then $\eta_j(t) = \eta(t)$ for all $t \in N(x_j)$. Clearly, $\eta_j$ is a solution to ${\cal I}_j$. As argued in Lemma~\ref{lem:course1P}, $\eta_j$ is the unique solution to ${\cal I}_j$. The algorithm given above also generates the unique solution $\mu_j=\eta_j$ for each ${\cal I}_j$ such that $j \in [n]$. Thus, due to our construction of $\mu$, it is the same as $\eta$. Hence, we return the matching $\mu$. 

In the reverse direction, if we return a matching $\mu$ which is a solution, then clearly, each sub-instance ${\cal I}_j$ for $j \in [n]$ has a solution $\mu_j$ (else, we directly return {\sf NO}) and is a yes-instance of \mefe. If ${\cal I}_j$ is a yes-instance, then $\mu_j$ must be both feasible and should lead to satisfaction for $x_j$. Therefore, each course should be satisfied with its allocation in $\mu$ as well. There must be no envy among TAs in matching $\mu_j$ for sub-instance ${\cal I}_j$, since $\mu_j$ is a solution to ${\cal I}_j$. Additionally, there can be no envy between TAs with different neighbouring courses. Hence, there is no envy between any TA in $\mu$. Therefore, $\mu$ is a \mefem matching and ${\cal I}$ is a yes-instance of \mefe.

\end{proof}
This completes the proof.
\end{proof}

\polyconstantcoursecap*
\begin{proof}
In this case, we argue that the total number of feasible matchings are at most $m^{\OO(1)}$, where $m$ is the number of TAs. Thus, we can try all possible feasible matchings and return the one that is a solution. If no matching meets our fairness criteria, then we return {\sf NO}. 
Note that the total number of TAs that get assigned in a matching is a constant value $c=\sum_{x_i \in X} c_i$. Given the set of TAs $T$, there are $\binom{m}{c}$ ways of choosing the set of matched TAs. For each of those ways, we need to further choose $c_i$ TAs for each course $x_i$. The number of such choices for each course $x_i$ are bounded by $\binom{c}{c_i}$, which is a constant given that $c$ is a constant. Since the total number of courses is constant, the product of all the possible choices for courses given a constant $c$ will also be a constant. Therefore, the total number of possible assignments is $\binom{m}{c} \cdot \OO(1)$ which is $m^{\OO(1)}$. 
\end{proof}

\begin{proof}[{\bf Proof of Theorem~\ref{thm:2val}}]
    Here, we prove the correctness. In particular, we prove the following result: 

\begin{restatable}{lemma}{lem:twotype}
    ${\cal I}$ is a yes-instance of \mefe if and only if ${\cal J}$ has a women-saturating stable matching. 
\label{lem:twotype}\end{restatable}

\begin{proof}
    In the forward direction, let $\mu$ be a solution to ${\cal I}$. We construct a matching $\eta$ for the instance ${\cal J}$ as follows. Consider the course $x_i\in X$. Let ${\sf sort}(\mu^{-1}(x_i))$ contains the TAs in $\mu^{-1}(x_i)$ in the decreasing order of their grades in $x_i$. Recall that no two TAs have the same grade in any course. Now, we consider men in $\mu^{-1}(x_i)$ in the order they appear in ${\sf sort}(\mu^{-1}(x_i))$ and they are matched to most preferred unmatched woman in $W_i\cup W'_i$. Since we always match an unmatched man to an unmatched woman, it is a matching. Let this matching be $\eta$. Next, we argue that $\eta$ is a stable matching for $\cal{J}$. Consider an unmatched pair of man and woman, say $(m,w)$. 
    Let $w$ be a woman corresponding to the course $x_i$. Due to our construction, corresponding to every course $x_i$, we have $c_i$ women. Since $\mu$ is a solution to ${\cal I}$ and every man in $\mu^{-1}(x_i)$ is matched to a woman corresponding to the course $x_i$, every woman is matched in $\eta$. If $m$ is unmatched in $\eta$, then TA $m$ is unmatched in $\mu$. 
    If $w$ prefers $m$ over $\eta(w)$, then, due to the construction of the instance ${\cal I}$, $g_{\eta(w)}(x_i)<g_{m}(x_i)$. Since $m$ is unmatched in $\mu$, $m$ envies $\eta(w)$, a contradiction. Consider the case when $m$ is matched. Then, 
    suppose that $m$ prefers $w$ over $\eta(m)$. Suppose that $w, \eta(m) \in W_i$. Due to the construction of $\eta$, $\eta(w)$ is before $m$ in ${\sf sort}(\mu^{-1}(x_i))$. Thus, $g_{\eta(w)}(x_i)>g_{m}(x_i)$. Thus, $w$ prefers $\eta(w)$ more than $m$. Hence, $mw$ is not a blocking pair. Similarly, we can argue when $w, \eta(m) \in W'_i$. Suppose that $\eta(m)$ is a woman corresponding to $x_j$. Then, $\eta(m)\in W_j \cup W'_j$. Since $m$ prefers $w$ over $\eta(m)$, TA $m$ prefers course $x_i$ over $x_j$. Then, if $w$ prefers $m$ over $\eta(w)$, $g_m(x_i)>g_{\eta(w)}(x_i)$. Hence, TA $m$ envies TA $\eta(w)$, a contradiction to the fact that $\mu$ is a solution to ${\cal I}$.

    In the reverse direction, let $\eta$ be a stable matching in ${\cal J}$ that saturates all women. We construct a matching $\mu$ for ${\cal I}$ as follows. If a woman $w$ corresponding to a course $x_i$ is matched to a man $m$, then for TA $m$, $\mu(m)=x_i$. Since $\eta$ saturates all women, and corresponding to every course $x_i$, we have $c_i$ women, $\mu$ is a feasible matching. Furthermore, due to the construction of preference lists no course/TA is matched to zero-valued TA/course. Next, we argue that for every course $x_i$, ${\sf AvgUtil}(x_i)\geq k$. Recall that due to the construction, $a_i$ many $q_i$ valued TAs are matched to $x_i$ and $a'_i$ many $q'_i$ valued TAs are matched to $x_i$. Thus, due to Equation~\ref{eq:1}, ${\sf AvgUtil}(x_i)\geq k$. Next, we argue that there is no envy between any pair of  TAs. Towards the contradiction, suppose that $t$ envies $t'$. Then, $g_t(\mu(t'))>g_{t'}(\mu(t'))$ and $u_t(\mu(t'))>u_t(\mu(t))$. Recall that we have a man $t$ corresponding to every TA $t$. Then, due to the construction, man $t$ prefers woman $\eta(t')$ more than $\eta(t)$ and woman $\eta(t')$ prefers man $t$ more than $t'$. Thus, $(t,\eta(t'))$ is a blocking pair, a contradiction.
\end{proof}
This completes the proof.
\end{proof}

\section{Missing proofs of Parameterized Algorithms}

\FPTm*
\begin{proof}
Let ${\cal I}$ be a given instance of \mefe.
Note that every TA has $n+1$ options for the allocation, as it can also remain unassigned. Thus, the total number of possible matching is $(n+1)^m$. For every possible matching, we can check in polynomial time if it is a solution to ${\cal I}$. Since, for a yes-instance, the number of courses is at least the number of TAs, this algorithm runs in $\OO(m^m(n+m)^{\OO(1)})$ time. 

\end{proof}

\FPTn*

\begin{proof}
Let ${\cal I}$ be an instance of \mefe that satisfies the constraints in the theorem statement. Let $v_i^X\colon T \rightarrow \{q_{1i},\ldots,q_{r_ii}\}$, where $r_i$ is a constant. Let $a_{ji}$ be the number of TAs of value $q_{ji}$ assigned to the course $x_i$ in a solution, where $i\in [n], j\in [r_i]$. We guess the values of $a_{ji}$, for all  $i\in [n], j\in [r_i]$, that satisfies the following two constraints.

\begin{align}
\sum^{r_i}_{j=1}a_{ji}q_{ji} \geq kc_i \label{eq:5} \\
\sum^{r_i}_{j=1}a_{ji} = c_i \label{eq:6}
\end{align}

Note that we have $(c_i+1)^{r_i}$ choices  for vector $(a_{1i}, \ldots, a_{r_ii})$ for every course $x_i$. A vector is called \emph{valid}, if it satisfies Equation~\ref{eq:5} and \ref{eq:6}. Since the capacity of each course and $r_i$ is constant, we have constant choices for every course.  Hence, we have $c^n$ total choices, where $c$ is a constant.

For each combination of $n$ valid vectors (one for each course) $(a_{1i}, \ldots a_{r_ii},\ldots,a_{1n}, \ldots a_{r_nn})$, we create an instance of the {\sc Stable Matching} problem ${\cal J}$ as follows. The set of men is $M=T$, i.e., corresponding to every TA in $T$, we have a man in $M$. Corresponding to every course $x_i\in X$, we have a set of women $W_{ji}$ that contains $a_{ji}$ many women. If any of these sets are empty we ignore those sets. Let $W_i$ be set of all women corresponding to course $x_i$. Due to the validity of the vector, corresponding to every course $x_i$, we have $c_i$ women. 
Next, we define the preference list of every woman $w$, say $P_w$, as follows. For every $w\in W_{ji}$, man $t$ is in $P_w$ if the valuation of course $x_i$ for the TA $t$ is $q_{ji}$. Recall that in an instance the {\sc Stable Matching problem} a man $m$ is in the preference list of a woman $w$ if and only if $w$ is in the preference list of $m$. Thus, woman $w$ is in $P_m$, the preference list of $m$, if and only if $m$ is in $P_w$.   
Next, we define the ordering of men in the preference list of every woman, which is based on the grades. Let $w$ be a women corresponding to the course $x_i$. Since no two TAs have the same grade for a course, for men $t,t'$ in $P_w$, $w$ prefers $t$ more  than $t'$ if and only if $g_t(x_i)>g_{t'}(x_i)$. Next, we define the ordering of women in the preference list of every man. 
Consider a man $t$. If $u_t(x_i) > u_t(x_j)$, where $x_i,x_j \in X$, then man $t$ prefers women in $w\in W_i$ more than woman in $\hat{w} \in W_j$, where $w, \hat{w} \in P_t$. If $w,w' \in P_t \cap W_i$, then we first note that there exists unique $j\in [r_i]$ such that $w,w' \in W_{ji}$, due to the construction. In this case, man $t$ orders them arbitrarily. 
This completes the construction of ${\cal J}$. If any of the constructed instance of {\sc Stable Matching} returns a women-saturating stable matching, then we return ``yes'', otherwise ``no''.

Next, we prove the correctness, which is similar to the proof of Lemma~\ref{lem:twotype}. In particular, we prove the following result:

\begin{restatable}{lemma}{lem:correctness-FPT}
 ${\cal I}$ is a yes-instance of \mefe if and only there exists an instance ${\cal J}$ that has a women-saturating stable matching, where ${\cal J}$ is one of the instances constructed above.
\label{lem:correctness-FPT}\end{restatable}

\begin{proof}
    In the forward direction, let $\mu$ be a solution to ${\cal I}$. Let $a^\star_{ji}$ be the number of TAs matched to $x_i$ in $\mu$ who are valued $q_{ji}$ by $x_i$. Clearly, $(a^\star_{1i},\ldots, a^\star_{r_ii})$, for each $i\in [n], j\in [r_i]$, is a valid vector. Since we try all possible options for $a_{ji}$, we also tried $a^\star_{ji}$.   Consider the instance ${\cal J}$ which is constructed for the values $a^\star_{ji}$, where $i\in [n], j\in [r_i]$.   We construct a matching $\eta$ for the instance ${\cal J}$ as follows. Consider the course $x_i\in X$. Let ${\sf sort}(\mu^{-1}(x_i))$ contains the TAs in $\mu^{-1}(x_i)$ in the decreasing order of their grades in $x_i$. Recall that no two TAs have the same grade in any course. Now, we consider men in $\mu^{-1}(x_i)$ in the order they appear in ${\sf sort}(\mu^{-1}(x_i))$ and they are matched to their most preferred unmatched woman in $W_i$. Since we always match an unmatched man to an unmatched woman, it is a matching.  Let this matching be $\eta$. 
    Next, we argue that $\eta$ is a stable matching for $\cal{J}$ that saturates all women. Consider an unmatched pair of man and woman, say $(m,w)$. 
    Let $w$ be a woman corresponding to the course $x_i$. Due to our construction, corresponding to every course $x_i$, we have $c_i$ women. Since $\mu$ is a solution to ${\cal I}$ and every man in $\mu^{-1}(x_i)$ is matched to a woman corresponding to the course $x_i$, every woman is matched in $\eta$. If $m$ is unmatched in $\eta$, then TA $m$ is unmatched in $\mu$. 
    If $w$ prefers $m$ over $\eta(w)$, then, due to the construction of the instance ${\cal I}$, $g_{\eta(w)}(x_i)<g_{m}(x_i)$. Since $m$ is unmatched in $\mu$, $m$ envies $\eta(w)$, a contradiction. Consider the case when $m$ is matched. Then, 
    suppose that $m$ prefers $w$ over $\eta(m)$. Suppose that $w, \eta(m) \in W_i \cap P_m$. Due to the construction of $\eta$, $\eta(w)$ is before $m$ in ${\sf sort}(\mu^{-1}(x_i))$. Thus, $g_{\eta(w)}(x_i)>g_{m}(x_i)$. Thus, $w$ prefers $\eta(w)$ more than $m$. Hence, $(m,w)$ is not a blocking pair. Then, $\eta(m)\in W_j \cap P_w$. Since $m$ prefers $w$ over $\eta(m)$, TA $m$ prefers course $x_i$ over $x_j$. Then, if $w$ prefers $m$ over $\eta(w)$, $g_m(x_i)>g_{\eta(w)}(x_i)$. Hence, TA $m$ envies TA $\eta(w)$, a contradiction to the fact that $\mu$ is a solution to ${\cal I}$.

    In the reverse direction, let  $\eta$ be a stable matching in ${\cal J}$ that saturates all women. We construct a matching $\mu$ for ${\cal I}$ as follows. If a woman $w$ corresponding to a course $x_i$ is matched to a man $m$, then for TA $m$, $\mu(m)=x_i$. Since $\eta$ saturates all women, and corresponding to every course $x_i$, we have $c_i$ women, $\mu$ is a feasible matching. Furthermore, due to the construction of preference lists no course/TA is matched to zero-valued TA/course, and due to the construction and validity of the vector, the satisfaction criteria of all the courses are met. Next, we argue that there is no envy between any pair of  TAs. Towards the contradiction, suppose that $t$ envies $t'$. Then, $g_t(\mu(t'))>g_{t'}(\mu(t'))$ and $u_t(\mu(t'))>u_t(\mu(t))$. Recall that we have a man $t$ corresponding to every TA $t$. Then, due to the construction, man $t$ prefers woman $\eta(t')$ more than $\eta(t)$ and woman $\eta(t')$ prefers man $t$ more than $t'$. Thus, $(t,\eta(t'))$ is a blocking pair, a contradiction.
\end{proof}

This completes the proof of theorem.
\end{proof} 

\thmfptapxn*

\begin{proof}
Let ${\cal I}$ 
be an instance of \mefe that satisfies the constraint in the theorem statement. Let $V_i=\max_{t\in T}v_i(t)$, i.e., maximum value a course assigns to a TA. Our basic idea is that for appropriately chosen $\epsilon'$, we guess the number of TAs assigned to course $x_i$ that have valuations in the range $[(1+\epsilon')^{j-1},(1+\epsilon')^j)$, for every course $x_i\in X$. 
Then, we create the number of copies of a course accordingly and reduce to the {\sc Stable Matching} problem  as in Theorem~\ref{thm:FPTn}. Note that for every course $x_i$, the number of guesses is $\OO(c_i^{\log_{1+\epsilon'}V_i})$. Thus, the total number of guesses for all the courses is at most $\OO({\sf cap}^{n\log_{1+\epsilon'}{\sf max_{val}}})$. Since the {\sc Stable Matching} problem can be solved in the polynomial time, the running time follows. The approximation guarantee is due to the fact that in a solution instead of assigning a TA of value $(1+\epsilon')^j$, the algorithm may assign a TA of valuation $(1+\epsilon')^{j-1}$, which incurs a loss of $\frac{1}{1+\epsilon'}$. We choose $\epsilon'=\frac{1}{1-\epsilon}-1$\footnote{In the short version, we wrote $\epsilon'=\frac{c}{1-\epsilon}-1$, however, this value also works well.}.

Next, we present the algorithm formally and argue the approximation factor. 
    
    For each course $x_i\in X$, we guess the number of TAs, say  $c_i^j$, that are assigned to a course $x_i$, in a solution, and $x_i$ has valuation at least $(1+\epsilon')^{j-1}$ and at most $(1+\epsilon')^j-1$ for these TAs, where $j\in[\log_{1+\epsilon'}V_i]$. A guess is call \emph{valid} if $c_i^1+\ldots+c_i^{\log_{1+\epsilon'}V_i}=c_i$. Note that we have at most $\OO({\sf cap}^{n\log_{1+\epsilon'}{\sf max_{val}}})$ valid guesses. Now, for every valid guess, $G=\{c_i^j \colon i\in [n], j\in [\log_{1+\epsilon'}V_i]\}$, we create an instance ${\cal I}_G$ of the {\sc Stable Matching} problem as follows. The set of men $M=T$. Next, we define the set of women. Corresponding to every course $x_i\in X$, we have a set of women $W_{ji}$ that contains $c_i^j$ many women. If any of these sets are empty, we ignore those sets. Let $W_i$ be set of all women corresponding to course $x_i$. 
    
    Next, we define the set of men who are in the preference list $P_w$ of a woman $w$. If the course $x_i$ has a valuation in the range  $[(1+\epsilon')^{j-1},(1+\epsilon')^j)$ for the TA $t$, then the man $t$ is in the preference list of woman $w\in W_{ji}$, otherwise not. The ordering of men and women in each other's preference list is defined in the same way as we did in Theorem~\ref{thm:FPTn}. We write here again for completeness. 
    The ordering of men in the preference list of every woman is based on the grades. Since no two TAs have the same grade for a course, for men $t,t'\in P_w$, $t$ is more preferred than $t'$ if and only if $g_t(w)>g_{t'}(w)$. Next, we define the ordering of women in the preference list of every man. Let $P_t$ be the set of women in the preference list of the man $t$. Consider a man $t$.  If $w,w' \in P_t \cap W_i$, then we first note that there exists unique $j\in [\log_{1+\epsilon'}V_i]$ such that $w,w' \in W_{ji}$, due to the construction. In this case, man $t$ orders them arbitrarily. 
 
This completes the construction of ${\cal I}_G$. For every ${\cal I}_G$, we find a stable matching $\mu_G$. If $\mu_G$ does not saturate all women, then we return ``no''. Otherwise, corresponding to every $\mu_G$, we construct a matching $\eta_G$ as follows: if a man $t$ is matched to a woman $w$ corresponding to the course $x_i$, then $\eta_G(t)=x_i$. If we obtain a matching $\eta_G$ that is feasible, merit-based envy-free, and for every course $x$, ${\sf Avg Util}(x)\geq (1-\epsilon)k$, we return it; otherwise we return ``no''. Next, we prove the correctness.  

\begin{lemma}\label{clm1:fptapx}
    Suppose that ${\cal I}$ is a yes-instance of \mefe. Then, there exists a guess $G$ such that ${\cal I}_G$ ha a women-saturating stable matching $\mu_G$. Let $\eta_G$ be the matching obtained as described above. Then, $\eta_G$ is a   feasible, merit-based envy-free, and for every course $x$, ${\sf Avg Util}(x)\geq (1-\epsilon)k$. 
\end{lemma}

\begin{proof}
 Except for the argument for satisfaction, the rest of the arguments are similar to the one in the proof of Lemma~\ref{lem:correctness-FPT}.  Let $\eta$ be a solution to ${\cal I}$. Let $\tilde{c}_{i}^j$ be the number of TAs matched to $x_i$ in $\eta$ who are valued in the range $[(1+\epsilon')^{j-1},(1+\epsilon')^j)$, where $j\in [\log_{1+\epsilon'}V_i]$, by $x_i$. Since we tried all possible options for $c_i^j$, we also tried $\tilde{c}_i^j$. Consider $G=\{\tilde{c}_i^j \colon i\in [n], j\in [\log_{1+\epsilon'}v_i]\}$.  Consider the instance ${\cal I}_G$.  We construct a matching $\mu_G$ for the instance ${\cal I}_G$ as follows. Consider the course $x_i\in X$. Let ${\sf sort}(\eta^{-1}(x_i))$ contains the TAs in $\eta^{-1}(x_i)$ in the decreasing order of their grades in $x_i$. Recall that no two TAs have the same grade in any course. Now, we consider men in $\eta^{-1}(x_i)$ in the order they appear in ${\sf sort}(\eta^{-1}(x_i))$ and they are matched to their most preferred unmatched woman in $W_i$. Since we always match an unmatched man to an unmatched woman, it is a matching.  Let this matching be $\mu_G$. 
    Next, we argue that $\mu_G$ is a stable matching for ${\cal I}_G$ that saturates all women. Consider an unmatched pair of man and woman, say $(m,w)$. 
    Let $w$ be a woman corresponding to the course $x_i$. Due to our construction, corresponding to every course $x_i$, we have $c_i$ women. Since $\eta$ is a solution to ${\cal I}$ and every man in $\mu^{-1}(x_i)$ is matched to a woman corresponding to the course $x_i$, every woman is matched in $\mu_G$. If $m$ is unmatched in $\mu_G$, then TA $m$ is unmatched in $\eta$.  
    If $w$ prefers $m$ over $\mu_G(w)$, then, due to the construction of the instance ${\cal I}$, $g_{\mu_G(w)}(x_i)<g_{m}(x_i)$. Since $m$ is unmatched in $\mu_G$, $m$ envies $\mu_G(w)$, a contradiction. Consider the case when $m$ is matched. Then, 
    suppose that $m$ prefers $w$ over $\mu_G(m)$. Suppose that $w, \mu_G(m) \in W_i \cap P_m$. Due to the construction of $\mu_G$, $\mu_G(w)$ is before $m$ in ${\sf sort}(\eta^{-1}(x_i))$. Thus, $g_{\mu_G(w)}(x_i)>g_{m}(x_i)$. Thus, $w$ prefers $\mu_G(w)$ more than $m$. Hence, $(m,w)$ is not a blocking pair. Then, $\mu_G(m)\in W_j \cap P_w$. Since $m$ prefers $w$ over $\mu_G(m)$, TA $m$ prefers course $x_i$ over $x_j$. Then, if $w$ prefers $m$ over $\mu_G(w)$, $g_m(x_i)>g_{\mu_G(w)}(x_i)$. Hence, TA $m$ envies TA $\mu_G(w)$, a contradiction to the fact that $\eta$ is a solution to ${\cal I}$.

    Let $\mu$ be a women-saturating stable matching of ${\cal I}_G$. Let $\eta_G$ be the matching constructed using $\mu$ as described above. Since $\mu$ saturates all women, and corresponding to every course $x_i$, we have $c_i$ women, $\eta_G$ is a feasible matching, because due to the construction of preference lists no course/TA is matched to zero-valued TA/course. Next, we argue that $\eta_G$ is merit-based envy-free.  Towards the contradiction, suppose that $t$ envies $t'$. Then, $g_t(\eta_G(t'))>g_{t'}(\eta_G(t'))$ and $u_t(\eta_G(t'))>u_t(\eta_G(t))$. Recall that we have a man $t$ corresponding to every TA $t$. Then, due to the construction, man $t$ prefers woman $\mu(t')$ more than $\mu(t)$ and woman $\eta(t')$ prefers man $t$ more than $t'$. Thus, $(t,\mu(t'))$ is a blocking pair, a contradiction. Next, we argue the satisfaction of every course. Due to the construction of ${\cal I}_G$ and $\eta_G$, for every $x_i \in X$, $\tilde{c}_i^j$ TAs are matched to $x_i$ whose valuations are in the range $[(1+\epsilon')^{j-1},(1+\epsilon')^j)$. Thus, for every course $x_i$, ${\sf AvgUtil(x_i)}\geq \sum_{j\in [\log_{1+\epsilon'}V_i]}\tilde{c}_i^j(1+\epsilon')^{j-1}$ with respect to matching $\eta_G$ and at most  $\sum_{j\in [\log_{1+\epsilon'}V_i]}\tilde{c}_i^j(1+\epsilon')^{j}$ with respect to $\eta$. Since ${\sf AvgUtil}(x_i)\geq k$ for every course $x_i\in X$, $k\leq \sum_{j\in [\log_{1+\epsilon'}V_i]}\tilde{c}_i^j(1+\epsilon')^{j}$, for every $i\in [n]$. Thus, ${\sf AvgUtil}(x_i)\geq \frac{k}{1+\epsilon'}$ for every course $x_i\in X$ with respect to the matching $\eta'$. Since $\epsilon'=\frac{1}{1-\epsilon}-1$, we obtained that ${\sf AvgUtil}(x_i)\geq (1-\epsilon)k$ or every course $x_i\in X$ with respect to the matching $\eta'$. 
\end{proof}
This completes the proof.

\end{proof}

\section{Missing proofs of Existence Section}

\binvalTAs*
\begin{proof}

We begin by constructing a bipartite graph $G'_{{\cal I}} = (Y, T)$, where $Y$ contains $c_i$ copies corresponding to every $x_i\in X$, i.e., $Y=\{x_i^1,\ldots,x_i^{c_i}:x_i\in X\}$. If an edge $x_it\in E(G_{\cal I})$, then $x_i^jt\in E(G'_{\cal I})$, for all $j\in [c_i]$. We find a maximum sized matching $\eta$ in  $G'_{{\cal I}}$ using the algorithm given in~\cite{MaxMatching73}. Accordingly, we create a matching $\mu$ for ${G_{\cal I}}$, by matching TA $t$ assigned to $x_i^j$ for $j \in [c_i]$ in $\eta$ to $x_i$ in $\mu$.  
\begin{lemma}
    $\mu$ is a feasible matching to ${\cal I}$.
\end{lemma}

\begin{proof}

Due to the construction of $G_{\cal I}$ and $G'_{\cal I}$, $t$ is not matched to any $0$-valued course. Next, we argue that $|\mu^{-1}(x_i)|=c_i$, for all $i\in [n]$. Towards this, it is sufficient to argue that $\eta$ saturates $Y$. We will prove it using Hall's Theorem. That is, we argue that for every $Y' \subseteq Y$, $|N(Y')|\geq |Y'|$. 
   
   Note that we can express $Y'$ as $\bigcup_{x_i \in S} \{ x_i^j \mid j \in J_i \}$, for some $J_i \subseteq [c_i]$ and $S \subseteq X$. Thus, $|Y'| \leq \sum_{x_i \in S} c_i$. 
   If $x_i \in S$, then $N(x_i) \subseteq N(Y')$. Thus, $N(Y') = \bigcup_{x_i \in S} N(x_i)=N(S)$. By constraint (iii), 
   $|N(S)| \geq \sum_{x_i \in S} c_i$ as $S\subseteq X$. Thus, $N(Y')=N(S)\geq \sum_{x_i \in S} c_i \geq |Y'|$. Hence, the maximum matching $\eta$ saturates $Y$.    
\end{proof}

Note that $\mu$ might not be \mefem. We modify $\mu$ using \texttt{Exchange Matching} (Algorithm~\ref{alg:existence-binval}) to obtain an \mefem matching.
The intuition of the algorithm is as follows. The algorithm first finds whether there is any pair of TAs $t_i, t_j$ s.t. $t_i$ has merit-based envy towards $t_j$ in a feasible matching $\mu'$, where $t_j$ is assigned to course $x_j$. The algorithm then finds the TA with the lowest grade matched to $x_j$ in $\mu'$, denoted by $t$, and exchanges the allocations of $t$ and $t_i$. The algorithm repeats this procedure until there is no such pair $t_i$ and $t_j$. 

\begin{algorithm}[]
\caption{\texttt{Exchange Matching}}\label{alg:existence-binval}
\textbf{Input:} an input instance ${\cal I}$ and a feasible matching $\mu : T \rightarrow X \cup \phi$ \\
\textbf{Output:} a \mefe $\mu'$.

\begin{algorithmic}[1]
\State $\mu' = \mu$
\State $T_M = \{t \in T : \mu'(t) \neq \phi \}$ \Comment{Set of Matched TAs}
\State $T_N = T \setminus T_M$ \Comment{Set of Unmatched TAs}
\While{there exists a pair $t_i, t_j$ such that $t_i \in T_N, t_j \in T_M$, $g_{i}(\mu'(t_j)) > g_{j}(\mu'(t_j))$ and $u_i(\mu'(t_j)) = a$} \label{step:loopt} 
    \State Chose such a pair $t_i, t_j$ arbitrarily
    \State $x_j = \mu'(t_j)$
    \State Find $t \in \mu'^{-1}(x_j) \text{ such that } g_t(x_j) \leq g_{t'}(x_j) \forall t' \in \mu'^{-1}(x_j)$ 
    \State $\mu'(t) = \phi$ and $\mu'(t_i) = x_j$ 
    \State $T_N = (T_N \setminus t_i) \cup \{t\}$
    \State $T_M = T_M \setminus t \cup \{t_i\}$
\EndWhile

\Return {$\mu'$} 
\end{algorithmic}
\end{algorithm}

It can be observed that if the algorithm terminates, then the resulting matching $\mu$ is merit-based envy-free since no unmatched TA envies any matched TA at that point (else, the while loop does not terminate) and no TA matched to some positively valued course will envy any other TA.

We now prove that the algorithm {\texttt{Exchange Matching}} necessarily leads to satisfaction of courses, and subsequently, that it terminates in polynomial time.

\begin{restatable}{lemma}{lem:satisfaction}
    {\texttt{Exchange Matching}} leads to satisfaction of courses.   
\label{lem:satisfaction}\end{restatable}

\begin{proof}
We recall the definition of $k_j = g_{t}(x_j)$, where $t = r_j^{-1}(c)$, with $r_j : T \rightarrow [m]$ (rank function), $c = c_1 + \ldots + c_n$ and $g_t(x_j) = v_j(t)$, the utility of the $x_j$ for $t$. Note that if every course $x_i$ is matched to only TAs in the set $\{r_i^{-1},\ldots,r_i^{-1}(c)\}$, then ${\sf AvgUtil}(x_i)\geq k_i \geq k$. We prove this using contradiction. Let us suppose that there is a TA $t$ assigned to a course $x_j$ in the matching output by the algorithm, $\mu'$, such that $r_j(t) > c$. Since $\mu'$ is feasible, the total number of TAs assigned to all courses is equal to $c$. Therefore, by the pigeonhole principle, there exists a TA $t'$ such that $r_j(t') \leq c$, which is not assigned to any course, since $t$ is assigned to $x_j$ while having a rank larger in number than $c$. But then, the pair $t, t'$ must satisfy the condition in the while loop of algorithm, which is a contradiction, since the algorithm terminated. 
\end{proof}

\begin{restatable}{lemma}{lem:termination}
    {\texttt{Exchange Matching}} terminates in polynomial time.   
\label{lem:termination}\end{restatable}
\begin{proof}
To prove that the algorithm terminates, we define two functions. For the set of all matchings, $\mathcal{M} : T \rightarrow X \cup \{\emptyset\}$, we first define a function $\sigma_j : \mathcal{M} \rightarrow Z_{\geq 0}$ corresponding to every course $x_j \in X$. For a matching $\mu$ of the  instance ${\cal I}$, let  $t\in \mu^{-1}(x_j)$ be a TA of lowest grade among all the TAs in $\mu^{-1}(x_j)$. 
Let 
    $E_{j, \mu} = \{t' : t' \in T \setminus \mu^{-1}(x_j) \text{ and } g_{t'}(x_j) > g_{t}(x_j)\}$
  and   
    $\sigma_j(\mu) = |E_{j, \mu}|$.

Observe that $\sigma_j(\mu)$ can never be negative, since it's value is governed by the size of a set.
We also define a potential function, $\psi : \mathcal{M} \rightarrow Z_{\geq 0}$, where 
    $\psi(\mu) = \sum_{i=1}^{n} \sigma_i(M)$

We observe that if $\psi(\mu)$ is 0 (i.e, $\sigma_j(\mu) = 0, \forall x_j \in X$), there cannot be any merit-based envious pair in $\mu$, since for every course $x \in X$, there is no TA which positively values $x$ and is not matched to it, while having a higher grade than any TA matched to $x$. We also observe that $\sigma_j(M)$ is bounded by the number of TAs, $|T|$.

It is now sufficient to prove that the exchange of TAs in each iteration of the while loop in {\texttt{Exchange Matching}} always leads to a strict decrease in the value of the potential function to show that the algorithm terminates. Consider an iteration of the while loop. Suppose that in this iteration 
TAs $t_i$ and $t$ are matched to and unmatched from course $x_j$, respectively. The initial matching before the exchange is denoted by $\eta$ and the subsequent matching is denoted by $\eta'$. 
Observe that in such a case, $\sigma_i(\eta) = \sigma_i(\eta')$ for $x_i \neq x_j$, since $\eta^{-1}(x_i) = \eta'^{-1}(x_i)$. We observe that the lowest grade of a TA matched to $x_j$ strictly increases, since we remove the TA $t$, which had the lowest grade, $g_{t}(x_j)$, among TAs matched to $x_j$ in $\eta$, and replace it with a TA $t_i$, with a strictly higher grade $g_{t_i}(x_i)$. TA $t_i$ would have been a part of $E_{j, \eta}$, since $g_{t_i}(x_j) > g_{t}(x_j)$, but $t$ is not a part of $E_{j, \eta'}$, since it has a lowest grade than all TAs matched to $x_j$ in $\eta'$. And there will be no addition of any other TAs in $E_{j, \eta'}$ since the lowest grade has increased. Therefore, $\sigma_j(\eta') < \sigma_j(\eta)$ and given that $\sigma_i(\eta) = \sigma_i(\eta')$ for $x_i \neq x_j$, $\psi(\eta') < \psi(\eta)$. Hence, the potential function always decreases with the exchange in each iteration of the while loop. Since the potential function must be non-negative, {\texttt{Exchange Matching}} must terminate.
\end{proof}

Therefore, {\texttt{Exchange Matching}} gives a \mefem matching for ${\cal I}$.
\end{proof}

Our next positive existence result is due to the fact that an instance of \mefe that satisfies constraints in Theorem~\ref{thm:thmexist} can be reduced to \hrpfull (\hrp).   

An instance I of \hrpfull involves a set of residents $R=\{ r_1,...,r_{n} \}$  and a set  of hospitals $H= \{h_1,...,h_m \}$. Each hospital $h_j \in H$ has a positive integral capacity $c_j$. Each hospital $h_j \in H$ has a preference list in which it ranks $R$ in strict order. Similarly, each resident $r_i \in R$ has a preference list in which they rank $H$ in strict order. A hospital is under-subscribed in a matching $\eta$ if $| \eta(h_j) | < c_j$. A matching $\eta$ is valid if $\eta(h_j) \leq c_j$ for each $h_j \in H$ and $| \eta(r_i) | \leq 1$ for each $r_i \in R$. 
 A pair $(r_i,h_j)$ blocks a matching $\eta$ if $r_i$ is unassigned or prefers $h_j$ over $\eta(r_i)$ and $h_j$ is under-subscribed or prefers $r_i$ over at least one resident in $\eta(h_i)$.
 A matching is said to be stable if it does not contain any blocking pair. In \hrp we are required to find a matching that is valid and stable. It is known that a valid stable matching in \hrp always exists as  mentioned in Theorem 3.2 of ~\cite{David2013}. The theorem proof is for the more general case of \hrp with ties and applicable to our case where we do not have ties.

\thmexist*

\begin{proof}

We create an instance ${\cal J}$ of \hrp  as follows. The set of residents is $R=T$, i.e., corresponding to every TA in $t_i\in T$, we have a resident $r_i\in R$ and the set of hospital is $H=X$, i.e., corresponding to every course in $x_i\in X$, we have a hospital in $h_i\in H$. The capacity of $h_i$ is equal to capacity of $x_i$, for all $i\in [n]$. The preference of resident is according to the preference of TAs, i.e., if $u_t(x_i) > u_t(x_j)$, then resident $r_t$ will prefer $h_i$ over $h_j$, and the preference of hospital is according to the grades of TA that is if $g_p(x_i) > g_q(x_i)$, then the hospital $h_i$ prefer $r_p$ over $r_q$. Since every TA gives distinct valuation to courses and every course gives distinct grades/marks to TAs all the preferences are complete and strict. Next, we prove the equivalence of two instances.

\begin{restatable}{lemma}{lem:lemexist}

${\cal I}$ is a yes-instance of \mefe if and only if ${\cal J}$ is yes instance of \hrp.
    
\label{lem:lemexist}\end{restatable}

\begin{proof}

In the forward direction, let $\mu$ be a solution to ${\cal I}$. We claim that $\eta = \mu$ is a solution to ${\cal J}$. Since the capacity of each course is met in $\mu$ and capacity of courses and hospitals are equal, $\eta$ satisfies the capacity constraints of ${\cal J}$. For contradiction let us assume that there is a blocking pair in $\eta$, i.e., there exist a pair $h_i$ and $r_j$ such that $h_i$ prefer $r_j$ over at least one resident in the set $\eta(h_i)$ and $r_j$ is either
unassigned or prefers $h_i$ over $\eta(r_i)$ in that case due to our reduction we have $t_p$ such that $t_p \in \mu^{-1} (x_i)$ , $g_{j}(x_i) > g_{p}(x_i)$ and $u_{j} (x_i) > u_{j} (\mu(t_j))$ which means $t_j$ envies $t_p$, which contradicts the fact that $\mu$ is a solution to ${\cal I}$.

In the reverse direction, let $\eta$ be a solution to ${\cal J}$. We claim that $\mu = \eta$ is a solution to ${\cal I}$.

Recall that throughout the paper, we are assuming that the number of TAs are greater than or equal to sum of capacity of courses. Thus, the number of resident is greater than or equal to the sum of capacity of hospital. Since, the preferences are complete, we know that all the hospital are saturated, else we have a blocking pair an under subscribed hospital and an unassigned resident. Hence, the capacity constraint of all the course are met. Since all the courses and TAs value each other positively,
${\sf AvgUtil}(x_i)\geq 1$, for all $i\in [n]$. 
Next, we argue that $\mu$ is merit-based envy-free. Towards the contradiction, suppose that $t_j$ envies $t_p \in \mu^{-1} (x_i)$. Then, $g_{j}(x_i) > g_{p}(x_i)$ and $u_{j} (x_i) > u_{j} (\mu(t_j))$. This implies that $h_i$ is matched to $r_p$ in $\eta$, 
and $r_j$ is either unassigned or prefers $h_i$ over $\eta(r_i)$ which contradicts the fact that $\eta$ is a stable matching for ${\cal J}$.
\end{proof}

As demonstrated in \cite{David2013} Theorem 3.2, the \hrp always has a solution hence \mefe with condition mentioned in Theorem~\ref{thm:thmexist} will always have a solution.
\end{proof}

\end{document}